%% file: main.tex
\documentclass[acmsmall]{acmart}

\usepackage{multirow}
\usepackage{pgfplots}
\usepackage{algorithm}
\usepackage{algpseudocode}
\usepackage{color}
\usepackage{xcolor}
\usepackage{listings}
\usepackage{todonotes}
\usepackage{array}
\usepackage{xparse}
\usepackage{hhline}
\usepackage{graphicx,pifont}
\usepackage{wrapfig}

\setcopyright{cc}
\setcctype{by-nc-nd}
\acmDOI{10.1145/3797141}
\acmYear{2026}
\acmJournal{PACMSE}
\acmVolume{3}
\acmNumber{FSE}
\acmArticle{FSE014}
\acmMonth{7}
\acmSubmissionID{fse26maina-p262-p}
\received{2025-09-11}
\received[accepted]{2025-12-22}
\acmMonth{7} 

\usepackage{amsmath,amsfonts}

\usepackage{amssymb}

\newcommand*\circled[1]{\tikz[baseline=(char.base)]{
    \node[shape=circle,draw,inner sep=2pt] (char) {#1};}}

\definecolor{codegreen}{rgb}{0,0.6,0}
\definecolor{codegray}{rgb}{0.5,0.5,0.5}
\definecolor{codepurple}{rgb}{0.58,0,0.82}
\definecolor{backcolour}{rgb}{0.95,0.95,0.92}

\lstdefinestyle{mystyle}{
    backgroundcolor=\color{backcolour},   
    commentstyle=\color{codegreen},
    keywordstyle=\color{magenta},
    numberstyle=\tiny\color{codegray},
    stringstyle=\color{codepurple},
   basicstyle=\ttfamily\tiny,
    breakatwhitespace=false,         
    breaklines=true,                 
    captionpos=b,                    
    keepspaces=true,                 
    numbers=left,                    
    numbersep=3pt,                  
    showspaces=false,                
    showstringspaces=false,
    showtabs=false,                  
    tabsize=2
}

 \pdfoutput=1 

\pgfplotsset{width=10cm,compat=1.9}

\AtBeginDocument{%
  \providecommand\BibTeX{{%
    \normalfont B\kern-0.5em{\scshape i\kern-0.25em b}\kern-0.8em\TeX}}}

\usepackage{cleveref}
\crefformat{section}{\S#2#1#3}
\crefname{figure}{Figure}{Figures}
\crefname{table}{Table}{Tables}
\crefname{listing}{Listing}{Listings}
\crefname{theorem}{Theorem}{Theorems}
\crefname{thm}{Theorem}{Theorems}
\crefname{lemma}{Lemma}{Lemmata}
\crefname{equation}{Eqt.}{Eqts.}

\newif\ifDEBUG

\DEBUGtrue

\ifDEBUG
\newcommand{\SR}[1]{\todo[color=lightgray,inline]{Sazz says: #1}}
\newcommand{\AB}[1]{\todo[color=yellow,inline]{Saleh says: #1}}
\newcommand{\BO}[1]{\todo[color=pink,inline]{Bosu says: #1}}

\else 

\newcommand{\SR}[1]{}
\newcommand{\AB}[1]{}
\newcommand{\BO}[1]{}
\fi

    \newenvironment{boxedtext}
    {    
    \begin{center}
    
    \begin{tabular}{|p{0.96\linewidth}|}
    \hline
    }
    { 
    \\ \hline
    \end{tabular} 
    
    \end{center}
       }

\definecolor{coolblack}{rgb}{0.0, 0.18, 0.59}

\newcommand{\cmark}{\ding{51}}%
\newcommand{\xmark}{\ding{55}}%

\definecolor{background_gray}{gray}{0.85}
\definecolor{background_green}{rgb}{7,163,90}

\definecolor{MidnightBlue}{rgb}{0.1, 0.1, 0.44}

\newcommand{\code}[1]{{\tt #1}}

\newcolumntype{L}[1]{>{\raggedright\let\newline\\\arraybackslash\hspace{0pt}}m{#1}}
\newcolumntype{C}[1]{>{\centering\let\newline\\\arraybackslash\hspace{0pt}}m{#1}}
\newcolumntype{R}[1]{>{\raggedleft\let\newline\\\arraybackslash\hspace{0pt}}m{#1}}

\begin{document}

\title{GraphQLify: Automated and Type Safety-Preserving GraphQL API Adoption}

\author{Saleh Amareen}
\orcid{0009-0001-8490-161X}
\affiliation{%
  \institution{Wayne State University}
  \city{Detroit}
  \country{USA}
}
\email{saleh.amareen@wayne.edu}

\author{Arif Rahman}
\orcid{0009-0009-5895-998X}
\affiliation{%
  \institution{Wayne State University}
  \city{Detroit}
  \country{USA}
}
\email{arif.shariar@wayne.edu}

\author{Sazzadur Rahaman}
\orcid{0000-0002-1258-6470}
\affiliation{%
  \institution{University of Arizona}
  \city{Tuscon}
  \country{USA}
}
\email{sazz@cs.arizona.edu}

\author{Amiangshu Bosu}
\orcid{0000-0002-3178-6232}
\affiliation{%
  \institution{Wayne State University}
  \city{Detroit}
  \country{USA}
}
\email{amiangshu.bosu@wayne.edu}

\renewcommand{\shortauthors}{Amareen \em{et} al.}

\input{Sections/abstract}



\maketitle
\sloppy

\input{Sections/introduction}

\input{Sections/related_work}

\input{Sections/design-choices}

\input{Sections/methodology}

\input{Sections/evaluation}

\input{Sections/results}
\input{Sections/discussion}

\input{Sections/limitations}

\input{Sections/conclusion}

\section*{ACKNOWLEDGEMENT}
\label{acknowledge}
This research is partially supported by the US National Science Foundation under Grant No. 2340389. 
The findings of this research do not necessarily reflect the views of the National Science Foundation. 

\section*{DATA AVAILABILITY}
\label{Data_availability}

The dataset, source code, and experimental results are publicly available on GitHub at \href{https://github.com/WSU-SEAL/graphqlify-fse-2026}{https://github.com/WSU-SEAL/graphqlify-fse-2026}.

\bibliographystyle{ACM-Reference-Format}
\bibliography{references}

\end{document}

%% file: Sections/abstract.tex
\begin{abstract}
GraphQL provides a schema-based, strongly typed query language that enables highly efficient client-server communication. This paper introduces \textit{GraphQLify}, an automated framework designed to migrate existing REST APIs to GraphQL. Unlike prior approaches that rely on relational databases, resource description frameworks (RDF), or machine-parsable specifications, \textit{GraphQLify} leverages static source code analysis for precise type inference. This novel technique generates GraphQL schemas that guarantee end-to-end type safety, preserving a core advantage of adopting GraphQL. Furthermore, existing migration tools typically generate separate adapter servers, which introduce performance overhead via dynamic request binding and network latency. \textit{GraphQLify} eliminates this by generating an embedded server that directly invokes the underlying API code, significantly improving performance. We evaluated \textit{GraphQLify} on 834 APIs across nine popular open-source projects, where it successfully converted 100\% of the APIs with zero type mismatches. In contrast, the current state-of-the-art tool, OASGraph, exhibited a 3.5\% failure rate and a 42\% type mismatch rate on the same dataset. Finally, our performance evaluation demonstrates that for workflows requiring five sequential API calls, clients using \textit{GraphQLify} reduce data fetching time by a factor of 2 to 4 compared to their REST counterparts.

\end{abstract}

%% file: Sections/introduction.tex
\section{Introduction}
\label{sec:introduction}
Service-oriented architecture (SOA) and Microservice-oriented architecture (MOA) are two popular architectures for developing software-as-a-service (SaaS) applications~\cite{quina2021quality}. In these client-server oriented architectures, clients communicate with interconnected application services via Application Programming Interfaces (APIs) where REpresentational State Transfer (REST) and Simple Object Access Protocol (SOAP) APIs have long dominated the industry~\cite{replacerest,controlledexperiment, quina2021quality}. Despite their wide adoption, these API choices have their disadvantages. For instance, REST clients interact with various server resources at the granularity of endpoints. These endpoint definitions are rigid since they return fixed data structures and often result in clients receiving more data than necessary (over-fetching) or performing multiple endpoint calls to obtain the desired data (under-fetching).

GraphQL offers a powerful alternative to traditional APIs by resolving critical data-fetching inefficiencies. Architectures like REST and SOAP often necessitate a cascade of slow, dependent network requests for a client to assemble a complete data view. GraphQL confronts this directly by empowering clients to request all necessary nested and related resources in a single, efficient server trip~\cite{graphql_spec}.
This powerful querying is rooted in GraphQL's strongly-typed schema, which acts as a definitive contract between the client and server. By defining every available data type, the schema enforces strict validation on all operations~\cite{graphql_spec}. This inherent type safety guarantees that clients receive data in the exact structure they expect, preventing a typical class of integration errors among REST and SOAP API clients.

A recent study found GraphQL offers significant performance advantages over REST and SOAP, including faster response times (up to 5 times improvement) and data transfer size reduction (up to 38 times smaller) for complex queries~\cite{quina2023systematicmapping}. As GraphQL offers a more efficient, flexible, and powerful way to query and manipulate APIs, it is currently adopted in well-known technology companies such as GitHub, IBM, Amazon, Twitter, Airbnb, Shopify, and Netflix~\cite{controlledexperiment, replacerest, graphqlformaldef2, evolutionarytesting, graphqlfoundation, ontology}. 
A recent Gartner report predicts that at least 50\% of enterprises will be utilizing GraphQL in production by 2025~\cite{ibm_gartner_blog_post}. 

However, developers adopting or migrating existing APIs to GraphQL often encounter several challenges. 
First, GraphQL introduces a new approach to managing client-server interactions compared to traditional architectures such as REST. REST is resource-based, while GraphQL is query-based. This fundamental difference imposes a steep learning curve, requiring software organizations to invest time, effort, and resources. For example, during their adoption of GraphQL, GitHub developers found GraphQL difficult to learn, resource-demanding, and time-consuming~\cite{replacerest}. 

Second, since GraphQL is strongly typed, integrating diverse data sources, such as existing APIs, databases, and applications, requires type conversions. Developers may resort to upcasting permissive types such as $Object$ or $String$ to ease GraphQL adoption. However, improper type conversion during serialization/deserialization or dynamic class loading can introduce bugs or even security defects~\cite{yazdipour2020github,mcfadden2024wendigo}.

 Third, a GraphQL implementation warrants additional API design, implementation, testing, and maintenance resources, especially as it often operates in parallel with existing APIs to maintain backward compatibility~\cite{graphqlwrappers}. Hence, scaling the adoption of GraphQL to hundreds or possibly thousands of services is challenging, especially if performed manually. Due to these challenges, a resource-constrained organization may defer adoption despite knowing GraphQL's advantages.

To address these adoption challenges, researchers and practitioners have proposed automation approaches to aid the migration of existing APIs to GraphQL~\cite{farre2019graphql,taelman2018graphql,ontology,morphgraphql,ugql,prisma,hasura, postgraphile,graphqlmesh,dgraph,graphqlwrappers,swaggertographal,stepzen}. These solutions can be organized into three categories. 
The \textit{first category} targets creating a GraphQL server for relational databases as an alternative query language to fetch data~\cite{prisma,hasura,postgraphile,graphqlmesh,dgraph}. 
These techniques map tables and columns directly to GraphQL types and fields, where resolvers perform Create, Read, Update, Delete (CRUD) operations.  While such automatic conversion to GraphQL is convenient, they are not code-aware, i.e., they do not consider business logic, input validation, or data integrity checks. Hence, their applicability is limited.
The \textit{second category} targets generating a GraphQL server for the Resource Description Framework (RDF) graph database, using ontologies\footnote{Ontologies are formal descriptions of knowledge as a set of concepts within a domain and the relationships that hold between them.} to define the structure and relationships within the RDF data and semantic mappings to align these definitions with GraphQL types and fields to enable structured queries~\cite{farre2019graphql, taelman2018graphql,ontology,morphgraphql,ugql}. However, these approaches require manual schema extraction by domain experts.
Solutions~\cite{swaggertographal,graphqlwrappers,stepzen} from the \textit{third category} aim to automate the translation of traditional APIs such as REST to GraphQL, utilizing dynamic analysis of machine-parsable API specifications such as Swagger~\cite{swagger}, OpenAPI~\cite{openapi}, or JSON responses of endpoints. However, these techniques either necessitate maintaining third-party representations of APIs or performing dynamic analysis, which can lead to a lack of type safety. For example,  since the OpenAPI specification supports a limited set of data types for simplicity and interoperability across platforms, tools using it for translation perform best-effort dynamic detection of data types based on their response JSON payloads. Additionally, the separate GraphQL servers generated by these techniques work as adapters to existing APIs and, therefore, incur performance overhead due to dynamic request binding and network latency.

\vspace{3pt}
\noindent \textbf{Novelty:} We introduce a novel approach that automatically generates GraphQL schemas by statically analyzing the source code of existing APIs. This technique uses type inference to ensure end-to-end type safety, a primary advantage of GraphQL over REST~\cite{quina2023systematicmapping}. To the best of our knowledge, this is the first work to do so.
To demonstrate our method, we present GraphQLify (pronounced "graph-qualify"), an automated framework for Java-based API migration. GraphQLify automatically generates a complete GraphQL schema and its resolvers, seamlessly migrating an API to GraphQL without requiring any code modifications. By injecting the generated code directly into the target project, our framework avoids the performance overhead of inter-application communication. We chose Java for our prototype due to its prevalence in enterprise applications~\cite{enterprise-list}.
\textit{GraphQLify} is compatible with various API implementations and requires no code changes, facilitating minimal-effort GraphQL adoption at scale. It also supports automatic API schema and resolver updates in the source code as the application evolves, minimizing maintenance efforts.




\vspace{3pt}
\noindent \textbf{Evaluation:} To assess its performance, we evaluated the performance of \textit{GraphQLify} on 834 APIs from nine well-known Open-Source Software (OSS) projects. \textit{GraphQLify} completes end-to-end migration in all cases. We also manually tested 50 APIs from six projects by running their respective applications and sending client requests. We observed correct responses in all cases.
We compared \textit{GraphQLify}'s performance against \textit{OASGraph}, an open-source solution to generate GraphQL schema from OpenAPI-based REST servers. 
Out of 196 APIs from five OSS projects,  OASGraph failed in {3.5\%} cases compared to 0\% failures from \textit{GraphQLify}. Additionally, while \textit{GraphQLify}'s auto-generated schema matched original types in 100\% cases, we found 42\% mismatches for OASGraph-generated ones. Finally,  our evaluation suggests that for application workflows requiring five sequential API calls, clients can reduce data fetching time by a factor of 2 to 4 by using the GraphQLify-generated APIs over REST counterparts.

\vspace{2pt}
\noindent \textbf{Contributions:} We summarize our key research contributions as follows:

\begin{itemize}
   \item We have designed and validated, \textit{GraphQLify}, the first static analysis-based tool for automated REST-to-GraphQL  conversion. 
   \item The success of our approach establishes a new, empirically superior method for ensuring type safety in REST to GraphQL API conversion.

   \item We have designed an `intermediate type definition' that functions as a systematic translator between REST and GraphQL architectures. This translation is managed by a suite of algorithms that perform the necessary specification conversion and data parsing.

    \item To promote adoption, we have released plugins to make \textit{GraphQLify} compatible with REST, OpenAPI~\cite{openapi}, and Spring~\cite{spring}. We also provide guidelines on how to build plugins targeting other programming languages or technologies.
       \item We make \textit{GraphQLify} and its evaluation data publicly available at \href{https://github.com/WSU-SEAL/graphqlify-fse-2026}{GitHub} to facilitate adoption and further research in this direction.
\end{itemize}

\vspace{2pt}
\noindent \textbf{Organization:} 
The remainder of this paper is organized as follows. Section~\ref{sec:related_work} overviews GraphQL and discusses closely related works. 
Section~\ref{sec-overview-workflow} provides a brief overview of GraphQLify and its workflow. 
Section~\ref{sec:design-considerations} details research challenges in designing GraphQLify and our mitigation strategies.
Section~\ref{sec:methodology} and ~\ref{sec:evaluation} detail the components of \textit{GraphQLify} and present the results of its evaluation, respectively.
Section~\ref{sec:discussion} provides recommendations for developers, researchers, and framework builders.
Finally, section~\ref{sec:conclusion} concludes the paper.

%% file: Sections/related_work.tex
\section{Background}
\label{sec:related_work}
The following subsections provide a brief overview of GraphQL and closely related works.

\subsection{GraphQL Core Concepts}
Let's consider a scenario where a developer wants to obtain all the communication on a pull request on GitHub. Since such communication can occur at multiple places (i.e., at the issue level, review level, or changed files), we require three REST requests (under-fetching). Moreover, each of these requests will return unnecessary and duplicate information (over-fetching). On the other hand, we can obtain the exact information using one precise GraphQL query listed in Figure~\ref{fig:queryexample}, thereby significantly reducing data fetching latency by a factor of 2 to 3.  To comprehend GraphQL,  we need to understand the following four core concepts.

\noindent \emph{1. GraphQL Schema}: A \textit{GraphQL schema} defines a type system of operations, data objects, and types representing the API. It also establishes a contract between the client and the server, detailing exactly what queries can be executed and the response format.

\begin{figure}
  \centering
  \vspace{-14pt}
  \includegraphics[width=0.7\linewidth]{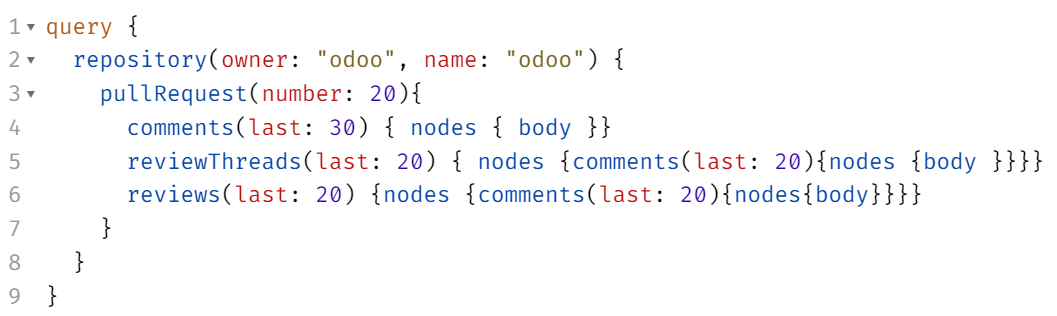}
  \caption{A GraphQL query to obtain all the discussions on a pull request. }
  \label{fig:queryexample}  
  \vspace{-12pt}
\end{figure}

\vspace{3pt}
\noindent \emph{2. GraphQL Query language} GraphQL's \textit{query language} allows clients to build requests to interact with various resources on the server. It supports three types of operations:  i) queries -- to let clients retrieve data from the server; ii)  mutations-- to support data modification; and iii) subscriptions-- to enable clients to subscribe to real-time events from a server.

\vspace{3pt}
 \noindent \emph{3. GraphQL Resolvers}: Resolver functions are handlers on the server side to fulfill the runtime execution of client requests to fetch, modify, or subscribe to resources. Resolvers may integrate with various data sources such as databases, REST services, or other APIs.

\vspace{3pt}
 \noindent \emph{4. GraphQL Lifecycle}:
GraphQL's query language parsing is built into the framework. Hence, an API designer defines a GraphQL schema and associated resolvers. A GraphQL request lifecycle starts when a client queries a GraphQL server endpoint via a query, mutation, or subscription (Figure~\ref{fig:queryexample}). The server validates the request against its schema. If invalid, the server returns an error, and no execution occurs. Otherwise, the server maps the requests or their nested fields to their resolvers for execution. Each resolver passes the request down to the next level of resolvers until no more requested fields remain. Lastly, the resolver returns the data to the GraphQL endpoint, which, in turn, returns the response, typically in JSON format.

\subsection{Related Works}
\noindent \textbf{Empirical Evaluation of GraphQL: } To characterize the benefits and challenges of GraphQL, Gleison \textit{et al.}~\cite{migratingtographql_apracticalassessment} conducted a practicality assessment of migrating from REST to GraphQL on seven systems. Their results indicate that GraphQL reduces the number of API calls and the overall payload size compared to REST. Sri \textit{et al.}~\cite{replacerest}  surveyed 38 GitHub employees to understand GitHub's GraphQL adoption process. They also assessed the benefits and challenges of REST and GraphQL in terms of efficiency and application feasibility, concluding that each protocol has its own benefits and weaknesses. Implementing a GraphQL server is difficult due to learning curves and requires a lot of time and resources~\cite{replacerest}. On the other hand, a controlled experiment by  Brito and Valente with 22 students found client-side application integration with a GraphQL API easier than with REST~\cite{controlledexperiment}.  They compared the implementation of remote service queries using REST and GraphQL. They found that GraphQL required less effort, with a median implementation time of 6 minutes compared to 9 minutes for REST. Additionally, implementing REST queries became more challenging with complex endpoints and multiple parameters. Participants without experience found GraphQL easier to implement and understand~\cite{controlledexperiment}. 
Wittern \textit{et} al.'s empirical investigation of over 8K GraphQL schemas mined from GitHub suggests that the majority of APIs have security issues that can be exploited using complex queries~\cite{wittern-2019}. 
Another direction focuses on assessing execution strategies and costs of GraphQL queries~\cite{cha-FSE-2020,mavroudeas-ASE-2021,roksela2020evaluating}.

\begin{table}[h]
\caption{Summary of comparison of \textit{GraphQLify} with prior art. Here, \cmark{},  \xmark{} indicate yes, and no, respectively.} 
\label{tab:tool-grid}
\setlength{\tabcolsep}{1pt}
\resizebox{0.7\linewidth}{!}{
\begin{tabular}{@{}lcccc@{}}
\toprule
Tools  & \rotatebox[origin=c]{60}{
\begin{tabular}{l}
     Proxy Adapter  \\
     Server Free? 
\end{tabular}} & \rotatebox[origin=c]{60}{Type Preserving?} & \rotatebox[origin=c]{60}{End-to-End?} &\rotatebox[origin=c]{60}{Tech. Agnostic?} \\ \midrule
MorphGraphQL~\cite{morphgraphql} &  \xmark{}  & \xmark{} & \xmark{} & \xmark{}                           \\
UltraGraphQL~\cite{ugql} &  \xmark{} & \xmark{} & \xmark{} &  \xmark{}                            \\
OBG-gen~\cite{ontology} & \xmark{} &  \xmark{} & \xmark{} & \xmark{}                             \\
StepZen~\cite{stepzen} & \xmark{}  & \xmark{} & \xmark{}&  \cmark{}                             \\
OASGraph~\cite{graphqlwrappers}  & \xmark{} &  \xmark{} & \cmark{} &  \xmark{}                            \\ 
GraphQLify (Ours)    &     \cmark{}  &   \cmark{}   & \cmark{} &   \cmark{}  \\ \bottomrule
\end{tabular}
}

    \end{table}

\vspace{2pt}
\noindent \textbf{Automated Tools to Support GraphQL Adoption:} 
Earlier automated GraphQL adoption approaches focused on generating GraphQL schemas and resolvers~\cite{ontology, morphgraphql, graphqlwrappers, ugql, farre2019graphql, taelman2018graphql}. \textit{Morph-GraphQL}~\cite{morphgraphql} requires semantics mapping as input and produces a GraphQL server as output. \textit{OBG-gen}~\cite{ontology} is an ontology-based solution that requires both ontology and semantics mapping as input to generate a GraphQL server. Both require domain expertise in declarative mapping languages and tools to create ontologies and semantics mapping rules, or manual setup of the prerequisites by domain experts. UltraGraphQL~\cite{ugql} does not require ontologies or semantics mappings but requires an RDF schema of SPARQL endpoints since RDF is schema-less. It also defaults all literals to \texttt{String}, leaving it susceptible to security injection attacks. \textit{OASGraph}~\cite{graphqlwrappers} targets REST-like APIs based on their machine-readable OpenAPI specification definitions~\cite{openapi}. However, it only works with OpenAPI and does not support any other format, and therefore, cannot support many existing non-OpenAPI applications. Moreover, since runtime user queries in \textit{OASGraph} are intercepted by a separate GraphQL adapter server (before mapping and delegating them to their respective REST-like APIs)~\cite{graphqlwrappers}, they become subject to network latency and mapping overhead. Table~\ref{tab:tool-grid} compares existing solutions against ours.
In terms of automated testing, researchers have focused on testing GraphQL servers using fuzzing~\cite{belhadi2024random}, genetic algorithms~\cite{evolutionarytesting}, queries harvested from production servers~\cite{zetterlund2022harvesting}, and access control evaluation using taint analysis~\cite{lambers2024taint}.


%% file: Sections/design-choices.tex
\section{\textit{GraphQLify}  Overview} 
\label{sec-overview-workflow}
We envision an automated solution for developers who want to add GraphQL support to existing APIs.
An ideal solution for this use case should i) handle the entire conversion process with minimal manual intervention; ii) preserve or enhance the type safety of the original API; and iii) offer simple and minimal effort setup.
Ideally, such a tool would take existing API code as input and automatically generate a fully functional GraphQL equivalent. 
Figure~\ref{fig:graphqlifySystem} provides a high-level overview of our design guided by this vision.

\subsection{Workflow} 
 \textit{GraphQLify} automatically generates a GraphQL schema from an API source code, where it takes an API source code as input and produces three artifacts: \textit{i)} a definition model, \textit{ii)} a GraphQL schema, and \textit{iii)} GraphQL resolvers. \textit{GraphQLify} has three modules, \textit{i)} \textit{Processor} (\S\ref{sec:processor}), \textit{ii)} \textit{Translator} (\S\ref{sec:translator}), and \textit{iii)} \textit{Generator} (\S\ref{sec:generator}). \textit{Processor} has three sub-modules, a Language-specific \textit{AST Parsers}, a Framework Semantic-aware \textit{API Processor}, and a \textit{Plugin Loader}. The Framework semantic-aware \textit{API processor} parses API source code and processes the API definitions using a framework-specific \textit{Plugin} and a Language-specific \textit{AST Parser}. We can add support for a new language by integrating its existing AST parser with \textit{API processor} through an adapter interface.
  An AST parser outputs a framework-specific representation of the API definition that can be processed by the plugins to create a definition model (i.e., \textit{Def Model}). Def Model is an intermediate representation we designed for the data types of the original API that need to be translated to GraphQL (\ref{fig:graphqlifyDefinition}). Subsequently, the \textit{Translator} component accepts this \textit{Def Model} as input to generate the corresponding GraphQL schema. The \textit{Generator} consumes both \textit{Def Model} and the GraphQL schema to produce resolvers.

\begin{figure*}
  \centering
  \includegraphics[width=\textwidth]{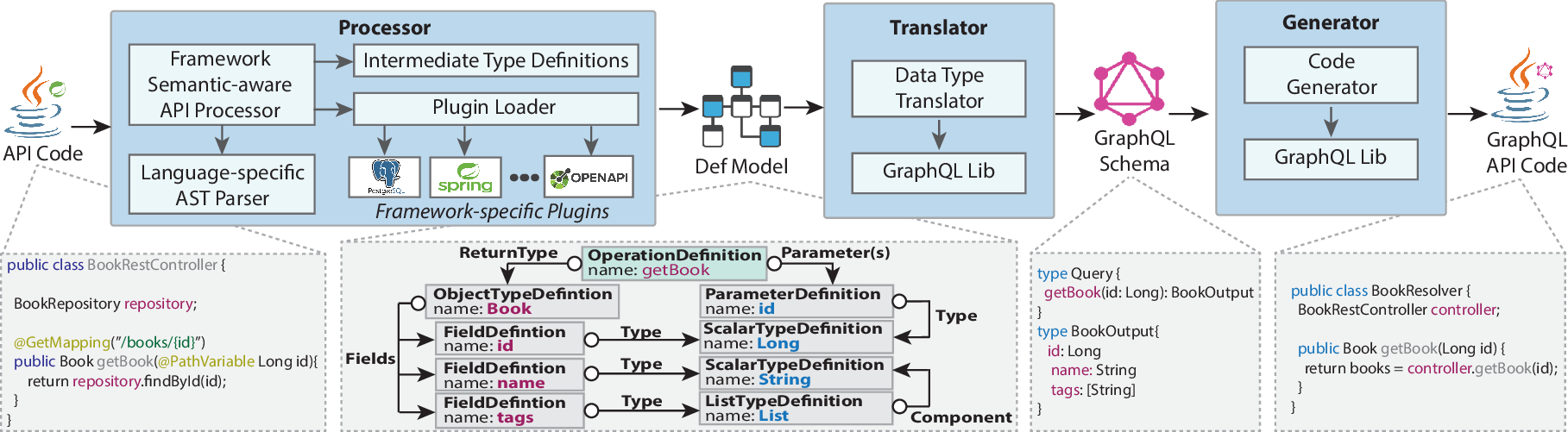}
  \caption{\textit{GraphQLify} Workflow. Plugins managed by the \textit{GraphQLify} processor identify APIs written with different frameworks within the input code and produce a definition model. Then, the  \textit{GraphQLify} translator translates the model to a GraphQL schema, which is used by the generator to generate API code written in GraphQL. Here, dashed arrows indicate \textit{uses} relationship.}
  \label{fig:graphqlifySystem}  
  \end{figure*}

\input{Sections/code-example}

\section{Design Goals and Challenges}
\label{sec:design-considerations}
The following subsections provide a detailed description of our design choices to implement \textit{GraphQLify} and corresponding research challenges encountered.

\subsection{Design Choices}
\label{sec:design-choices}
Our considerations and design choices to implement \textit{GraphQLify} are as follows.

\vspace{2pt}
\noindent \textit{1. Static Analysis-Based Processing} 
Static analysis-based tools support more versatile use cases compared to their dynamic-analysis counterparts because of their low setup cost~\cite{DBLP:conf/uss/MachirySCSKV17, DBLP:conf/ccs/RahamanXASTFKY19}. 
Hence, we employed the static analysis-based paradigm to design \textit{GraphQLify} to enable its potential to be used to convert APIs written in diverse technologies seamlessly. Specifically, \textit{GraphQLify} uses AST-based code analyses to identify existing API semantics and data types.

\vspace{2pt}
\noindent \textit{2. Preserving types by design}: Since type-safety is a key motivation for building microservice-based enterprise applications using GraphQL instead of REST~\cite{quina2023systematicmapping}, \textit{GraphQLify} prioritizes type preservation. Moreover, lack of type preservation can open up new attack surfaces~\cite{yazdipour2020github,mcfadden2024wendigo} and interoperability issues~\cite{ontology,seifer2019empirical,ugql}. \textit{GraphQLify} utilizes static code analysis to automatically recover types from API code to produce strongly typed definitions and comply with GraphQL's type-safety premise.

\vspace{2pt}
\noindent \textit{3. Decoupling type identification and mapping}: There are differences between programming language type systems and GraphQL specification~\cite{graphql_spec}. For example, GraphQL does not support sets, multiple or void return types, multi-dimensional tuples, maps, or dictionaries. On the other hand, GraphQL has separate types to indicate the distinction between input and output data. To address that, we implement an intermediate-type system (\ref{fig:graphqlifyDefinition}) to decouple GraphQL types from the programming languages in which the code is written.

\vspace{2pt}
\noindent \textit{4. Extensibility with the Plugin System}: \textit{GraphQLify} is designed to be a technology-agnostic system to enable migrating code from any programming language or web API technology. It uses a plugin-based design principle to enable extensions to various programming languages or API technologies (i.e., Spring~\cite{spring}, OpenAPI~\cite{openapi}).

\subsection{Key Challenges and Their Mitigations}
Key research challenges towards realizing the \textit{GraphQLify}'s design goals (\ref{sec:design-choices}) and our mitigation strategies are as follows.
    
\vspace{2pt}
\noindent \textbf{ $\triangleright$ Challenge 1: Accurate API Detection and Processing.} 
Given a source code repository as input, the first problem we need to tackle is detecting existing APIs and processing them.
The diversity of syntactic and semantic differences in existing technologies makes this a challenging problem. 
For example, Java's JAX-RS specification uses annotations e.g., \texttt{@GET}, \texttt{@PUT}, \texttt{@POST}, and \texttt{@DELETE} to define REST APIs~\cite{jaxrs}, where Spring uses \texttt{@GetMapping}, \texttt{@PutMapping} \texttt{@PostMapping}, \texttt{@DeleteMapping} to suggest GET, PUT, POST and DELETE methods, respectively~\cite{spring}. Additionally, the semantic meaning of an API code depends on the underlying framework. For example, \texttt{ResponseEntity<T>} in Spring supports generic types, whereas JAX-RS does not.

\textbf{$\checkmark$ Mitigation:} Even basic static analysis tasks are undecidable -- for instance, determining if a variable will ever be assigned a specific value~\cite{DBLP:journals/loplas/Landi92}. Thus, utilizing domain knowledge is critical in designing practical tools. Hence, \textit{GraphQLify} leverages the insight that most frameworks use API definitions to declare and enforce framework-specific functionalities. These API definitions work as a separation layer between the framework and the developer code. Therefore, \textit{GraphQLify} uses framework-relevant method definitions to extract API semantics, such as name, parameters, output types, and API implementations. Since method definitions are sufficient for API detection and processing, \textit{GraphQLify} skips analysis of implementations, which also helps avoid unsoundness or imprecision inherent to aliases, pointers, or data flow analysis~\cite{DBLP:conf/uss/MachirySCSKV17}. It decouples framework-specific, semantics-aware processing from the original design via a plugin-based extensibility mechanism.

\vspace{2pt}
\noindent \textbf{$\triangleright$ Challenge 2: Incompatibility Between REST and GraphQL Type System.} To handle incompatibility between programming language types and GraphQL types in a principled manner, \textit{GraphQLify} designs an intermediate type system.
\emph{Processor} maps type definitions from input APIs to the intermediate definitions in \textit{GraphQLify}'s intermediate type system, which is used by \emph{Translator} to generate a GraphQL schema. These intermediate types are also used by \emph{Generator} to map types from the GraphQL schema to the programming language while generating the \textit{resolvers}. However, designing an intermediate type system to maintain type safety during type conversion between incompatible types is challenging.
First, it is essential to establish a structured contract that all types in the system adhere to. This contract must support polymorphism and preserve hierarchical structures to ensure lossless type transformation. It should also address type mismatches between programming languages and the GraphQL schema.

\textbf{$\checkmark$ Mitigation:} We studied popular programming languages (i.e., Java, Python, JavaScript, PHP) for web development. We studied their basic (i.e., primitive or built-in) type systems to create similar representations in our type system. Specifically, the design of our intermediate type-system was guided by the principle that it must be: \emph{i)} \emph{type-safe} avoiding under- or over-casting (e.g., Double to Float or BigInt to Long), \emph{ii)} \emph{lossless} preventing loss of precision or mismatching types (e.g., mapping a DateTime to String), \emph{iii)} \emph{minimal and portable} defining only the essential types to allow portability to other programming languages despite our prototype implementation in Java. \emph{iv)} \emph{explicit over implicit} requiring explicit type definitions instead of relying on ambiguous mappings (i.e., mapping a \texttt{Address} complex type to an \texttt{Object}),\emph{v)} \emph{extensible} supporting domain-specific types via user-defined types and custom scalars, and \emph{vi)} \emph{preserving structure} ensuring that metadata such as type relationships, hierarchy, and constraints are not lost when moving between systems, such as supporting generic, nullable, and dictionary types.

\begin{figure}
\begin{lstlisting}[caption=Exampe of Java generic type arguments, label={lst:java_generics}]
    public class Response<T>{
        T payload;
        // other fields, getters, setters
    }
    
    // Simplified Java REST API Declarations
    @GET Response<User> getUser(int id);
    @GET Response<Address> getUserAddress(int id);
    @GET Response<Account> getUserAccount(int id);
\end{lstlisting}
\end{figure}

\noindent \textbf{$\triangleright$ Challenge 3: Identification of Types for Generics and Data Wrappers.} 
Source APIs may contain generic types. However, unlike Java or TypeScript, GraphQL does not support generic type arguments, which poses a challenge for static analysis-based translation. Examples include unwrapping REST-specific response wrappers such as \texttt{Response<T>}, where developers can pass different \textit{type} arguments in different APIs. Capturing different types distinctly is vital to ensure type-preserving API translation.

\textbf{$\checkmark$ Mitigation:} 
To preserve type safety during API conversion, GraphQLify must resolve generic wrappers, such as Response<T>. It operates on the insight that maintaining the specific type arguments used in these wrappers is sufficient to replicate the original code's type safety. For instance, consider the APIs in Listing~\ref{lst:java_generics} returning \textit{Response<User>}, \textit{Response<Address>}, and \textit{Response<Account>}. Even if these arguments have subclasses, preserving these top-level types is enough.
GraphQLify uses static analysis to detect each unique instantiation of a generic type and then constructs a corresponding static type definition. For the example in Listing~\ref{lst:java_generics}, it generates three distinct GraphQL types \textit{DataOfUser}, \textit{DataOfAddress}, and \textit{DataOfAccount}. This method ensures the converted API is compatible with GraphQL’s static typing, passes schema validation, and preserves type safety.

\vspace{2pt}
\noindent \textbf{$\triangleright$ Challenge 4: Namespace Conflicts.} 
\label{sec:language-features}
Programming languages allow duplicate names in different namespaces, where \textit{namespaces} are used to group related code entities and can be seen as a scope within which names are defined. However, GraphQL does not support namespaces. A GraphQL schema only separates read, write, and subscription operations into their respective root nodes. Thus, handling duplicate names for different entities poses a major challenge for semantic-preserving API conversion.
The issue is further exacerbated for object-oriented languages supporting method overloading or overriding, where such APIs use the same names.

\textbf{$\checkmark$ Mitigation:} Similar to as generic types, we utilize static analysis to identify interfaces, classes, and methods with duplicate names in different namespaces.
Then, we apply renaming strategies that incorporate the names of input arguments, return types, and namespaces to create different identities in the GraphQL domain. For example, to distinguish two \texttt{User} types from two different namespaces (i.e., \texttt{com.w.User} and \texttt{com.z.User}), \textit{GraphQLify} will use their fully-qualified names in the GraphQL schema, such as \texttt{com\_w\_User} and \texttt{com\_z\_User}. Similarly, to handle method overloading, it will append method return types and arguments to create separate identities. For example, it will rename the following two methods, i.e., \texttt{User get (Integer id)} and \texttt{User get (String name)}, to \texttt{getUsingIntegerReturnsUser} and \texttt{getUsingStringReturnsUser}.

%% file: Sections/code-example.tex
\subsection{Generated Artifacts and The Lifecyle of a GraphQLify-Generated API  }

Listing \ref{lst:rest-example} shows the source code of a REST API implemented in Spring REST~\cite{springrest}, one of the most popular frameworks to build REST services in Java-based enterprise applications. Here, the controller defines two REST endpoints: 1) \emph{getArticle}: to retrieve an existing blog article with an associated ID parameter; and 2) \emph{addArticle}: to create a new article. 
\begin{lstlisting}[caption=REST controller source code, label={lst:rest-example}]
    //ArticleController.java
    @RestController
    @RequestMapping("/api/articles")
    public class ArticleController {
        @Autowired
        private ArticleService articleService;

        @GetMapping("/{id}")
        public ResponseEntity<Article> getArticle(@PathVariable Long id) {
            Article article = articleService.getArticle(id);
            return new ResponseEntity< >(article, HttpStatus.OK);
        }
        
        @PostMapping
        @PreAuthorize("hasRole('USER')")
        public ResponseEntity<ArticleResponse> addArticle(@Valid @RequestBody ArticleRequest articleRequest, @CurrentUser UserPrincipal currentUser) {
            ArticleResponse articleResponse = articleService.addArticle(articleRequest, currentUser);
            return new ResponseEntity<>(articleResponse, HttpStatus.CREATED);
        }        
    }
\end{lstlisting}

Migration with \textit{GraphQLify} starts by simply integrating \textit{GraphQLify} into a project's build system (i.e., Maven~\cite{maven}). After that, during the build process, it automatically generates a complete GraphQL API layer from the existing REST services. Specifically, it creates the following key artifacts: 
\begin{itemize}
    \item A foundational GraphQL schema that defines the API's structure.
    \item A set of corresponding data model classes for all REST parameters and return types, which are placed in the {\tt graphqlify.generated.graphql.model} package.
    \item Three essential Java classes to manage runtime logic: \emph{QueryResolver.java} and \emph{MutationResolver.java} contain the definitions for GraphQL queries and data-modifying mutations, respectively (Listing~\ref{lst:resolver-example}). ObjectMapper.java handles the critical data mapping between the original REST types and their new GraphQL counterparts.
\end{itemize}
Once deployed, the fully-functional GraphQL API, including its schema and resolvers, becomes accessible at a standard GraphQL endpoint (e.g., /graphql).

\begin{lstlisting}[caption=Auto-generated resolvers and object mappers, label={lst:resolver-example}]
  // QueryResolver.java    
   public class QueryResolver implements GraphQLQueryResolver {
      ArticleController articleController;
      public Article getArticle(Long id, DataFetchingEnvironment dataFetchingEnvironment) throws Exception {
        com.sopromadze.blogapi.model.Article response = articleController.getArticle(id).getBody();
        return ObjectMapper.mapGetArticle(response);
    }
    // MutationResolver.java  
   public class MutationResolver {
      ArticleController articleController;
      public ArticleResponse addArticle(PostRequestInput postRequest, DataFetchingEnvironment dataFetchingEnvironment) throws Exception {
        ArticleRequest articleRequest_ = __mapPostRequestInput(postRequest, new java.util.HashMap<>());
        InjectedParameter currentUser_injected = new InjectedParameter("com.sopromadze.blogapi.controller.ArticleController", "addArticle", "currentUser", "com.sopromadze.blogapi.security.UserPrincipal", dataFetchingEnvironment);
        com.sopromadze.blogapi.security.UserPrincipal currentUser_mapped = injectedParameterMapper.mapInjectedParameter(currentUser_injected);
        response = articleController.addArticle(articleRequest_, currentUser_mapped).getBody();
        return ObjectMapper.mapAddArticleResponse(response);
    }
}

  //ObjectMapper.java
  import  graphqlify.generated.graphql.model.Article;
  import graphqlify.generated.graphql.model.ArticleResponse;
    public class ObjectMapper {
     public static Article mapGetArticle(com.sopromadze.blogapi.model.Article response) {
      Article graphql_response = new Article();
     //auto-generated mapping
     ....
     return graphql_response;
     }

    public static ArticleResponse mapAddArticleResponse(com.sopromadze.blogapi.payload.ArticleResponse response) {
      ArticleResponse graphql_response = new ArticleResponse();
     //auto-generated mapping
     ....
     return graphql_response;
     }
    }

\end{lstlisting}
When a client sends a request to the GraphQL endpoint, the underlying `graphql-java` library intercepts and processes it. The framework first validates the query against the schema and then executes the corresponding resolver method to fetch the data. For instance, consider the following:\\ {\tt 
query \{
  getArticle(id: "10") \{
    author
    title
  \}
\}}

To fulfill this request, `graphql-java` invokes the \emph{`getArticle`} method within the generated `QueryResolver.java` (Listing~\ref{lst:resolver-example}, line 4). This resolver acts as a bridge, calling the original REST API method (Listing~\ref{lst:rest-example}, line 9). The \emph{Article} object returned by the REST method is then mapped to its GraphQL-compatible `Article` model by a dedicated method in `ObjectMapper.java`.
Finally, this mapped object is passed back to the `graphql-java` library, which parses this object to generate the response. Crucially, the library ensures the final response only contains the fields explicitly requested by the client—in this case, just the `author` and `title`.

\textbf{Handling Authorization:}
The code generation process also accounts for cross-cutting concerns, such as security. If a resolver's underlying REST method requires authorization (e.g., `addArticle` in Listing~\ref{lst:rest-example}), GraphQLify automatically injects the necessary authentication details. It passes context-specific authorization parameters to the REST method using a customized `InjectedParameter` object, as shown in the generated resolver code (Listing~\ref{lst:resolver-example}, line 13).

%% file: Sections/methodology.tex
\section{GraphQLify System Components}
\label{sec:methodology}

This section details \textit{GraphQLify}'s system components and their designs.

\subsection{Intermediate Type Definition Model}
\label{sec:api_processing}

The processor component (\S~\ref{sec:processor}) analyzes APIs identified by the plugins to produce their intermediate definitions. To analyze an API, the processor must resolve its components, e.g., name, input parameters, declaring class, and output type, and then map them to our intermediate type definition. 
To support polymorphism, we define two root classes, namely \texttt{TypeDefintion} and \texttt{Definition} classes. \texttt{TypeDefintion} (like \texttt{Object} in Java) acts as the base for all regular data types, i.e., String, Integer, etc., whereas
\texttt{Definition} (like \texttt{Class} in Java) acts as the base for all meta element definitions, i.e., fields, parameters, and operations, etc.  Various other types under these two base definitions are as follows.

\begin{figure}
  \centering
  \includegraphics[width=0.8\linewidth]{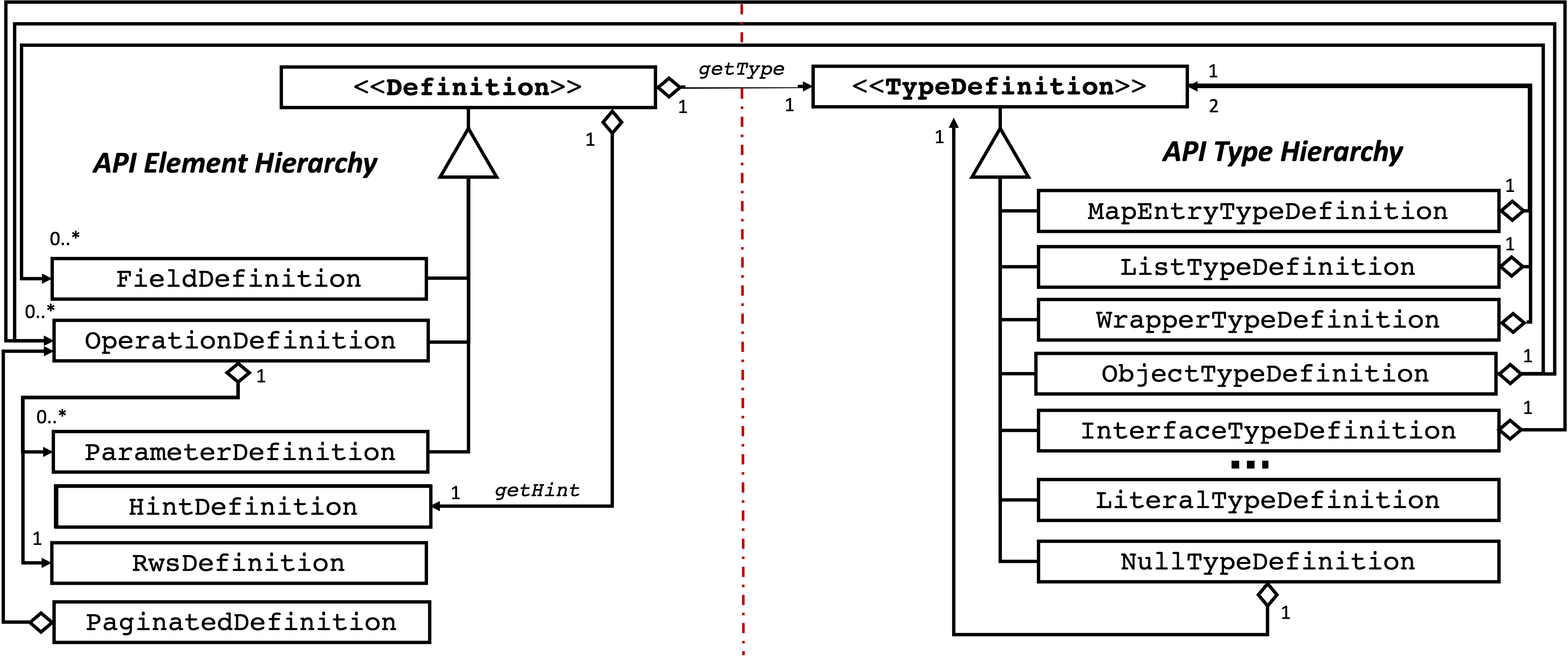}
  \caption{Main components of our proposed intermediate definition model. This definition model is used to convert types between REST and GraphQL. 
  }
    \label{fig:graphqlifyDefinition}
    \vspace{-15pt}
\end{figure}

\noindent
\textit{ \textbf{$\diamond$ Data Type Definitions.}} GraphQLify ensures type preservation during API translation using an intermediate type system, defined by the TypeDefinition hierarchy (Figure~\ref{fig:graphqlifyDefinition}). A dedicated Processor maps types from the source API to this intermediate model. To guarantee fidelity, the Processor prevents common mapping errors such as: i) Under- or over-casting (e.g., \texttt{Double} to \texttt{Float}); ii) Losing precision (e.g., mapping  \texttt{Address} type to a generic \texttt{Object}), and iii) Mismatching types (e.g., mapping \texttt{DateTime} to\texttt{String}). 
The important types in our intermediate model are as follows.

\vspace{2pt}
    \noindent \textit{\textbf{$\diamond$ Primitive Types:}} We map programming language built-in primitive data types (e.g., Integer, String) to a new type in our type system named \texttt{LiteralDef}. However, GraphQL supports only five scalar types, namely: Int, String, Boolean, Float, and ID. Therefore, we introduce two additional scalar categories: i) extended scalars allow framework developers to extend GraphQL scalars for programming language primitives not natively supported in GraphQL (e.g., \texttt{long} or \texttt{char}), and ii) custom scalars allow application programmers to specify user-defined scalars (e.g., \texttt{PhoneNumber}).

    \vspace{2pt}
    \noindent \textit{\textbf{$\diamond$ Void and Nulls:}} We designate a specific \texttt{VoidDef} for \textit{void}-like (i.e., return nothing) API output types, since GraphQL does not support void types. While defining void types as \texttt{LiteralDef} is achievable, dedicating a specific type definition provides flexibility concerning its GraphQL target definition. Also, by default, we assume all types are nullable and wrap non-null types (e.g., \texttt{String!}) in a \texttt{NonNullDef}.

    \vspace{2pt}
    \noindent \textit{\textbf{$\diamond$ List Type:}} For array-like types, GraphQL supports list types. Therefore, our intermediate definition model maps vectors, sets, lists, collections, arrays, or array-like types to an intermediate type named \texttt{ListDef}. We also recursively convert their component types to their respective type definitions, which allows \textit{GraphQLify} to build nested lists (i.e., lists of lists).

    \vspace{2pt}
    \noindent \textit{\textbf{$\diamond$ Map Type:}} GraphQL does not support maps, dictionaries, or multi-dimensional arrays. Therefore, we convert them to lists of entries, where each entry is a synthesized object of the components. Taking a key-value map as an example, we recursively map both key type \texttt{K} and value type \texttt{V} to construct an entry object type definition, i.e., \texttt{MapEntryDef}, containing the key and value type definitions. This results in a \texttt{ListDef} where the component type is \texttt{MapEntryDef}.

    \vspace{2pt}
    \noindent \textit{\textbf{$\diamond$ Object Type:}} Mapping complex types follows the following rules. If it is an interface, we recursively process its public methods, resulting in an \texttt{InterfaceDef}. For object types, we recursively map all non-transient (i.e., serializable) fields and map each one of them recursively to construct an \texttt{ObjectDef} with its fields. For parameterized complex types that leverage type arguments (e.g., \texttt{Response<String>}, we recursively process the type argument instances (if available) and attach them to their complex type's definition. This step is crucial to help uniquely name the parameterized complex type during the GraphQL translation step.

    \vspace{2pt}
    \noindent \textit{\textbf{$\diamond$ Wrappers:}} 
    In some APIs, certain wrapper types used to wrap the output types are superfluous and only bear significance in the original APIs or are considered transient. Taking \texttt{ResponseEntity<Customer>} as an example, the type \texttt{ResponseEntity} is a wrapper type that is only meaningful in REST APIs. Therefore, \textit{GraphQLify} provides the ability to define common wrapper types at the framework, plugin, or user levels as a configuration property. If such a type is encountered, the semantic processor unwraps and discards the outer type and only processes the inner type, i.e., \texttt{Customer}.

\vspace{3pt}
\noindent
\textbf{\textit{$\diamond$ Meta Type Definitions.}} To preserve type information for meta elements, i.e., fields, parameters, operations, etc., we define corresponding meta types  (type hierarchy under \texttt{Definition}). One notable meta definition is \texttt{OperationDefinition}, which is used to define operation types for GraphQL. In our intermediate type, we define three possible operation types: read, write, and subscribe, which we infer based on an operation's intent. We define a \textit{readonly} API as one that does not modify the state or cause side effects on the server and, thus, maps it to a \texttt{RwsDef} value of \texttt{READ}.
Otherwise, we presume an API is \textit{writable} and maps to \texttt{RwsDef} value of \texttt{WRITE}. For example, REST methods marked with \texttt{GET} HTTP designators are read-only and therefore map them to \texttt{READ}, whereas \texttt{POST}, \texttt{PUT}, \texttt{PATCH}, or \texttt{DELETE} are \textit{writable} and map them to \texttt{WRITE}. Similarly, \texttt{select} statements in database scripts are \textit{readonly}, whereas \texttt{update} statements are not. Lastly, understanding the semantics of subscription-based operations is complex; \textit{GraphQLify} currently only supports read and write.

\noindent
\textbf{Soundness and Completeness Guarantees.} One can view our intermediate type system as a theory of mapping types to types. We formalize this relationship using the judgment:  $\frac{
    t \text{ } \mapsto \text{ } \tau_{IR}
    \quad
    \tau_{IR}\text{ }\hookrightarrow \text{ }\tau
  }
  {
    \Gamma\text{ }\vdash \text{ } t \text{ }: \text{ }\tau
  }$, which reads as: \textit{under the intermediate type theory $\Gamma$, the source language type $t$ semantically maps with the GraphQL type $\tau$}, if there exists an IR type $\tau_{IR}$ such that (i) $t$ semantically maps ($\mapsto$) with $\tau_{IR}$, and (ii) $\tau_{IR}$ semantically maps ($\hookrightarrow$) with $\tau$.

\begin{lemma}
    (Soundness) The type theory $\Gamma$ is \emph{sound} for a source language $L$ if, for any type $t \in \mathsf{Type}_L$, $\Gamma \vdash t : \tau \implies semantic\{t\} \subseteq semantic\{\tau\}$. It means if our type theory $\Gamma$ maps a source language type to a GraphQL type, the mapping is semantically safe.
\end{lemma}

\textit{Proof Sketch.} Proving soundness of $\Gamma$ would boil down to the following cases.

\begin{itemize}
    \item \textit{Case 1: direct correspondence.} For a given $t$, $\exists \tau.\text{ } \Gamma \vdash t : \tau$, such that $t$ and $\tau$ has the same semantics. Examples of such types are primitives (i.e., \texttt{Integer}, \texttt{Boolean}, \texttt{String}, etc.)  and List types. For these cases, $\Gamma$ establishes an exact semantic mapping. For example, following is the type judgment for List:

\begin{center}
$\frac{
    \texttt{List}<v> \text{ } \mapsto \text{ } \texttt{ListDef}<v>
    \quad
    \texttt{ListDef}<v> \text{ } \hookrightarrow \text{ } [v]
  }
  {
    \Gamma \text{ } \vdash \text{ } List<v> \text{ } : \text{ } [v]
  }$.
\end{center}

Here, soundness follows trivially, since by construction, $\Gamma$ preserves semantics for both $t$ and $\tau$ in our intermediate type representation $\tau_{IR}$. Same argument holds for \texttt{Object} types too.

\item \textit{Case 2: no direct correspondence.} For a given $t$, $\nexists \tau. \text{ } \Gamma \vdash t : \tau$, such that $t$ does not have a semantically equivalent type $\tau$. For this case, $\Gamma$ creates a new representation $\tau_{IR}$, such that $semantic\{t\} \subseteq semantic\{\tau_{IR}\} \subseteq  semantic\{\tau\}$. Examples of such types include Map, which has no native support in GraphQL. We present the type judgement for Map as follows.

\begin{center}
$\frac{
    \texttt{Map}<k, v> \text{ } \mapsto \text{ } \texttt{ListDef}<\texttt{MapEntryDef}(k, v)>
    \quad
    \texttt{ListDef}<\texttt{MapEntryDef}(k, v)> \text{ } \hookrightarrow \text{ } [\texttt{MapEntryDef}(k, v)]
  }
  {
    \Gamma \text{ } \vdash \text{ } Map<k, v> \text{ } : \text{ } [\texttt{MapEntryDef}(k, v)]
  }$.
\end{center}
Here, too, soundness follows immediately as \texttt{ListDef}<\texttt{MapEntryDef}($k, v$)> preserves the semantics of the original \texttt{Map}<$k, v$> type. Same argument holds for \texttt{void} and \texttt{null} types too. 

\end{itemize}

Similarly, completeness implies $\Gamma$’s capability to handle all valid types  $t \in \mathsf{Type}_L$. As our intermediate type system is designed to handle types from popular programming languages, thus, completeness is guaranteed by construction.

\subsection{Processor: \textit{GraphQLify} API Processing}
\label{sec:processor}
We design a plugin-based mechanism to handle framework-aware AST processing.
A plugin targets APIs written in a specific technology, framework, or programming language. Plugins use AST-based parsers to parse input source code to an abstract syntax tree representation and process APIs in a fashion specific to the technology or framework they support. For instance, one plugin may target Java, specifically scanning for Spring controller methods representing REST APIs~\cite{spring} based on HTTP method designators in the source code (e.g., \texttt{@GET}, \texttt{@POST}, \texttt{@PUT}, \texttt{@PATCH} or \texttt{@DELETE}). Another plugin may scan API operations based on their OpenAPI~\cite{openapi} definitions. Similarly, others may target SQL database scripts scanning for annotations in their data definition language (DDLs), data modification language (DMLs), or stored procedures. 
After identifying all the data types and operation types (i.e., \texttt{GET}, \texttt{POST}, \texttt{PUT}) defined in the method declarations, \textit{GraphQLify} maps them with our intermediate types (\texttt{OperationDef}) to create a definition model.

\subsection{Translator: Definition Models to Schema} 
\label{sec:translator}
The translator module uses the definition model generated by the processor to create the GraphQL schema. It iterates over each operation definition (\texttt{OperationDef} in Figure~\ref{fig:graphqlifyDefinition}) and transforms them to a corresponding GraphQL operation. This process uses the mapping between our intermediate type definitions and GraphQL types, which we created as part of the intermediate type definition creation process. Then, we use \textit{GraphQL lib}~\cite{graphql_java} to validate the GraphQL definitions and, in turn, to create a GraphQL schema.

\subsection{Generator: Resolver Generation}
\label{sec:generator}
To produce a fully functional GraphQL server, the last step in \textit{GraphQLify} is to generate the GraphQL resolver functions that handle user runtime GraphQL queries. \textit{GraphQLify}'s generator module produces source code of the following: (i) an object model of GraphQL input and output types, (ii) GraphQL resolver functions of queries and mutations, (iii) compile-time mappings that bind GraphQL objects to invoke Original API code,  and (iv) GraphQL server runtime configuration. \textit{GraphQLify}'s three-step approach is as follows.


\vspace{2pt}
\noindent
\emph{Step 1: GraphQL type to code type conversion} The generation process starts with creating data and object types for code from the GraphQL schema. If the type of a GraphQL element is complex (i.e., object or interface), the Generator iterates its fields and builds their respective types recursively. Next, it generates a file containing the  GraphQL object type before adding it to the set of GraphQL objects. It also generates a separate file for enum types and adds them to the set of types. For other types, it uses the intermediate definition model to convert them into language-specific types to generate their source code type. For example, we convert \texttt{ListOfIntStringEntry} back to a \texttt{Map<Integer, String>} or a custom GraphQL scalar \texttt{Byte} to a Java \texttt{Byte} class.

\vspace{2pt}
\noindent
\emph{ Step 2: Build Resolvers} The Generator builds two resolver classes to handle GraphQL queries and mutations, respectively. It iterates all the operations of the \texttt{root} (i.e., query or mutation) from GraphQL schema to generate their respective resolver functions. For each operation, it maps the types of its arguments to their corresponding REST API inputs. Similarly, it maps the return type to the REST API outputs. Next, it generates the resolver function with the appropriate mapping and invokes the original API methods to handle the request. Finally, this resolver function is added to the set of resolver functions in the root resolver.

\vspace{2pt}
\noindent
\emph{Step 3: Build a GraphQL Server} The Generator also creates the configuration classes necessary to create a GraphQL server. This process includes (i) creating a web handler to respond to the \texttt{/graphql} endpoint, (ii) creating a schema parser that reads, parses, and validates the generated GraphQL schema, and (iii) wiring the generated GraphQL resolvers to handle users' queries.

\subsection{Implementation} 
\label{sec:impl}
We implemented our \textit{GraphQLify} prototype in Java to handle API codes written in Java, targeting JDK 17 compatibility. Our implementation has approximately 10K lines of Java code and 1K of configuration and build files.  We use the Java Annotation Processing (JAP) API~\cite{javaAnnotationProcessor} to parse annotations.
To generate Abstract Syntax Trees (ASTs) from code, \textit{GraphQLify} uses Spoon, a third-party Java library~\cite{spoon2006}. Spoon operates directly on the AST of the Java annotation processor and provides a strongly typed and easy-to-understand API that enables traversing the AST in a declarative fashion. To implement \textit{GraphQLify}'s schema generator, we use GraphQL-Java~\cite{graphql_java}. GraphQL-Java is widely recognized as the de facto standard for building GraphQL APIs in Java. Post-completion, we also use GraphQL-Java to validate the generated schema. 

\vspace{2pt}
\noindent \emph{Plugins}:
To validate \textit{GraphQLify}'s pluggability, we developed four independent plugins that target various technologies in the Java programming language, namely: \textit{openapi-plugin} targeting APIs defined using the OpenAPI specification~\cite{openapi}, \textit{jaxrs-plugin} for REST APIs implemented using JAX-RS~\cite{jaxrs}, \textit{springweb-plugin} for Spring-based Java REST APIs~\cite{spring}, and \textit{graphql-plugin} as a custom plugin. For example, the \textit{openapi-plugin} utilizes static analysis to detect Java APIs with OpenAPI semantics e.g., methods with OpenAPI's \texttt{@Operation} annotation, \textit{springweb-plugin} targets controller methods with Spring's semantics such as \texttt{@GetMapping} or \texttt{@PostMapping}.

\vspace{2pt}
\noindent \emph{Execution modes}:
Similar to OASGraph~\cite{graphqlwrappers}, \textit{GraphQLify} operates in two modes: (1) \code{strict} mode generates partial schema by skipping APIs that fail due to programming language features (\S~\ref{sec:language-features}) (2) \code{non-strict} mode automatically mitigate issues encountered during schema generation using mitigation heuristics detailed in \S~\ref{sec:language-features}.
The \code{non-strict} mode aims to achieve the envisioned end-to-end automation.  

\vspace{2pt}
\noindent \emph{Failure diagnosis}:
\textit{GraphQLify} includes a \textit{Failure diagnosis} to help developers identify the root causes in the Source API with suggestions to fix them. If GraphQLify fails to convert an API in the \code{strict} mode, it logs a warning message that contains the following information -- (i) \textit{Category} represents a high-level title of the warning, (ii) \textit{Description} provides additional human-readable details about the warning, (iii) \textit{StackTrace} specifies the root cause of the warning, (iv) \textit{Location} points to the source file and line number, and (v) \textit{Resolution Strategy} describes the approach \textit{GraphQLify} would have taken to mitigate a failing API in the \code{non-stict} mode.

%% file: Sections/evaluation.tex
\section{Empirical Evaluation} \label{sec:evaluation}

We evaluate \textit{GraphQLify} with nine popular open-source Java projects against OASGraph~\cite{graphqlwrappers}. The following subsections detail our evaluation metrics, configuration, and results.

\subsection{Project Selection}

We select nine open-source projects (Table~\ref{table:evaluation_projects}) from GitHub to evaluate \textit{GraphQLify} based on the following six criteria: 
    i) Open source project written in Java;
    ii) REST implementation uses OpenAPI, Spring Web, or JAX-RX;
    iii) Has at least 2,000 lines of code;
    iv) Has at least 10 REST APIs;    
    v) Has at least 50 stars; and
    vi) Has at least 50 forks.    
The first two criteria are selected since \textit{GraphQLify} requires a plugin to parse APIs from each REST implementation. In the current prototype, we implemented four plugins; three target the listed ones in the second criterion. 
We use the third and fourth criteria to exclude small-scale projects. Finally, we apply the remaining two criteria to select popular projects with community interests. 

We manually inspect a project to check its development activities and see if it meets all those criteria. 
We purposely seek source code that may not necessarily only contain friendly data types. This step is important to help test \textit{GraphQLify}'s ability to handle diverse APIs. For example, \textit{Rest Countries} includes data types that are not defined in the GraphQL specification (e.g., date and/or time) as well as data types that do not have any significance in a GraphQL schema, such as Java's \texttt{Object} type. We also selected projects to represent diverse application domains. 

\input{tables/evaluation-projects2}

\subsection{Evaluation Setup}
We evaluate \textit{GraphQLify} with a project based on the following eight steps.
\circled{1} We import our prototype artifact as a dependency to a project with the respective plugin to match the source code's API technology. Upon the compilation of each project, \textit{GraphQLify} automatically performs source-code processing to scan for qualifying APIs and generate a corresponding GraphQL schema. 
 \circled{2} To avoid inflating GraphQLify's performance by iteratively improving it to fit our test projects, we use  \hyperlink{https://github.com/r-spacex/SpaceX-API}{SpaceX APIs} as unit tests during development.
    \circled{3}  We run \textit{GraphQLify} with the \code{strict} mode and compute the five metrics listed in Section~\ref{sec:metrics}. We also record categories of failures and mitigation strategies reported by the tool.
     \circled{4} We run \textit{GraphQLify} with the \code{non-strict} mode and compute the same five metrics. 
    \circled{5} Two of the authors independently inspected each of the generated APIs and the source code of the original API to evaluate its correctness and type preservation. These two inspectors have extensive experience developing enterprise Java applications with microservices, RESTful, and GraphQL APIs used by millions of users.  
    As both inspectors agreed on all cases, we did not require a conflict resolution step.    
    \circled{6} If the project uses the OpenAPI specification, we also generate schemas using OASGraph, a current state-of-the-art tool that we compare against. The inspectors also manually inspected OASGraph-generated APIs against original APIs to evaluate type preservation.
    \circled{7} We deployed six projects on a server and manually added data using the application. Next, we manually test 10 APIs covering diverse use cases from each project by sending GraphQL queries using an Altair-based GraphQL client~\cite{altair} and comparing them against the expected response. This step further validates that our end-to-end automated migration works. \circled{8} Finally, to compare real-world data-fetching latency between REST and GraphQLify-generated APIs, we deployed five of the applications on Google Cloud Platform (GCP). We sent 5 inter-dependent queries to both REST and GraphQL APIs and measured round-trip time. 
Our evaluation answers five questions detailed in the following subsections.

\subsection{EQ1: How Does \textit{GraphQLify} Perform in its Two Modes?}
\label{sec:metrics}
For this evaluation, we select the following five evaluation metrics.
\circled{1} \emph{Failure rate} measures the percentage of APIs in an application that a tool fails to transform to GraphQL.
    \circled{2}  \emph{Invalid schema rate} represents the percentage of generated schema that failed validation~\cite{graphql_java}.
    \circled{3}  \emph{Schema Size} represents the lines of code in an auto-generated GraphQL schema.    
    \circled{4} \emph{Type Count} indicates the total number of types in an auto-generated GraphQL schema. 
    \circled{5} \emph{Type mismatch} refers to the cases where types in the generated schema differ from those in the source API.

Table~\ref{table:evaluation_projects} also shows the performance of \textit{GraphQLify} in two execution modes. In the \code{strict} mode, it fails for 31.3\% of the APIs. Notably, higher shares of failures are observed among large-scale, complex projects (i.e., Mall and Metas Fresh). This finding is not surprising since the potential of namespace conflict grows with the size of the project code. \textit{GraphQLify}'s \code{non-stict} mode can mitigate all these cases and has a 0\% failure rate. Both modes generate valid GraphQL schemas in all cases.  \textit{GraphQLify} generated schemas range between 85 and 6,000 lines, indicating the scalability of our solution.
Our manual evaluation found no type mismatch in both modes -- if the type is present in the source code, \textit{GraphQLify} matches.

\begin{table}
    \centering
    \caption{Evaluation of \textit{GraphQLify}'s heuristics based type mapping}
    \label{tab:correction-eval}
    \begin{tabular}{|p{3.4cm}|r|r|r|r|r|r|r|} 
    \hline
        \multirow{2}{*}{\textbf{Project}} & \multicolumn{3}{c|}{\textbf{Mapping}} & \multicolumn{4}{c|}{\textbf{Mapping causes}} \\ \hhline{~-------}
            &  \rotatebox[origin=c]{60}{\#Total}& \rotatebox[origin=c]{60}{\#Valid} &\rotatebox[origin=c]{60}{\# Invalid} & \rotatebox[origin=c]{60}{Missing} & \rotatebox[origin=c]{60}{Invalid} & \rotatebox[origin=c]{60}{Unknown} & \rotatebox[origin=c]{60}{Conflict} \\ \hline
Mall & 114 &114  & 0  & 1 & 0 & 0 & 113 \\ \hline
Metas Fresh & 165 & 165 & 0 & 0 & 35 & 35 & 95 \\ \hline
Exadel CompareFace & 10 & 10 & 0 & 0 & 6 & 0 & 4 \\ \hline
Rest Countries & 11 & 11 & 0 & 0 & 0 & 0 & 11\\ \hline
Blog REST API & 0 & NA & NA & 0 & 0 & 0 & 0\\ \hline
E-commerce REST API & 9 & 9 & 0 & 0 & 0 & 0 & 9\\ \hline
Pet Clinic REST & 0 & NA & NA & 0 & 0 & 0 & 0\\ \hline
Travels Java API & 4 & 4 & 0 & 0 & 0 & 0 & 4\\ \hline
Task Mgmt System & 5 & 5 & 0 & 0 & 0 & 4 & 1\\ \hline
\textbf{Total} & \textbf{318} & \textbf{318} & \textbf{0} & \textbf{1} & \textbf{41} & \textbf{39} & \textbf{237}\\ \hline
    \end{tabular}
    
    \end{table}

\subsection{EQ2: What is the Accuracy of  \textit{GraphQLify}'s Heuristics based Type Mapping?}

Table~\ref{tab:correction-eval} shows the performance of \textit{GraphQLify}'s heuristics-based type mapping. \textit{GraphQLify} made 318 type mappings to generate schemas in the \code{non-strict}. Our manual evaluation found all mappings to be valid, which shows their reliability.  
Our results also indicate that most of the required heuristics-based mappings (75\% ) are due to name conflicts, which is unsurprising since these applications are written in Java, which supports polymorphism-based designs. Types that violate the GraphQL specification (i.e., ones without any field) rank second with 13\%, followed by unknown types (12\%).

\subsection{EQ3: Do \textit{GraphQLify}-Generated APIs Return Correct Data?}

For this evaluation, we selected 50 GraphQL APIs (queries and mutations) from five projects, with 10 from each project. Verified projects include Rest Countries,  Task Mgmt System, Travels Java, Blog,  and Pet Clinic. For the remaining four projects, we had difficulty setting up the projects and running the REST APIs. Moreover, this step is time-consuming since it requires running the application, configuring API servers, and populating data through interactions. We take a convenient sample to ease configuration and data generation time. Before marking an API as correct, we validate it in several ways. First, we validate that the content of each query response matches that of its REST counterparts. Second, we attempt to send invalid data or fields to ensure the GraphQL server validates them against the schema. Third, we validate that the GraphQL client's introspection shows the GraphQL query and its input and output object structures. Finally, we validate sending multiple queries within the same request. Our validation suggests 100\%  end-to-end successful API conversions, as all 50 APIs pass the aforementioned criteria and return correct data.

\begin{table}
    \centering
    \caption{Failure rate and type mismatch of \textit{GraphQLify} compared to OASGraph~\cite{graphqlwrappers}}
    \label{tab:oasgraph}
    \begin{tabular}{|l|r|R{2.4cm}|r|R{2.4cm}|} \hline
   \multirow{2}{*}{\textbf{Project}} &  \multicolumn{2}{c|}{\textbf{OASGraph}} & \multicolumn{2}{c|}{\textbf{GraphQLify}} \\ \hhline{~----}
      & \textbf{Failure} & \textbf{Type mismatch} & \textbf{Failure} & \textbf{Type mismatch} \\ \hline
     Exadel CompareFace    & 7  & 14 & 0 &  0 \\ \hline
     E-commerce REST API    &  0 & 26  &  0&  0  \\ \hline
     Pet Clinic REST    & 0 &  6& 0 &  0 \\ \hline
     Travels Java API   & 0  & 32 & 0 & 0  \\ \hline
     Blog REST API     & 0 & 4 & 0 & 0 \\ \hline
     \textbf{Overall}   & \textbf{7 (3.5\%)}  & \textbf{82 (41.8\%)} &\textbf{ 0 (0\%)} &  \textbf{0 (0\%)}  \\ \hline
    \end{tabular}
    \end{table}
\subsection{EQ4: How Does \textit{GraphQLify}'s Performance Compare Against OASGraph?} 
For this evaluation, we only select OASGraph~\cite{graphqlwrappers} since it is the only one that provides end-to-end automation (Table~\ref{tab:tool-grid}), a necessity for independent evaluation on a third-party project. However, OASGraph only supports OpenAPI, and three of the nine projects do not comply with OpenAPI specifications. Although Metas Fresh complies with the OpenAPI specification, we failed to run it with OASGraph.
Table~\ref{tab:oasgraph} compares the performance of \textit{OASGraph} against \textit{GraphQLify} with the remaining five projects. We ran both tools in \code{non-strict} mode. While OASGraph has a failure rate of 3.5\% (i.e., 7 out of 196 APIs), \textit{GraphQLify} has 0\% failure. Although OASGraph has a low failure rate, it fails to preserve the original type in most cases (42\%). On the other hand, \textit{GraphQLify}'s generated schema fully matches the ones from the source code in all cases. These results also show the superiority of our approach, especially where type safety is a priority.

\subsection{EQ5: How much Data Fetching time does a Client Save with \textit{GraphQLify}-Generated API Compared to REST by Avoiding REST's Under-fetching Issue?} 
To answer this question, we designed a test consisting of five interdependent queries. In the REST setting, each query depended on the data returned by the previous one, thereby exposing the underfetching problem inherent to REST. For both the REST and \textit{GraphQLify}-generated endpoints, we measured the average round-trip time (RTT) across five separate trials to fetch the final dataset.
The results, illustrated in Figure~\ref{fig:graphqlexample}, show that \textit{GraphQLify}-converted APIs significantly outperform REST APIs in scenarios requiring multiple, dependent data fetches.
While GraphQL has a slight overhead on single requests due to its added computational overhead, its ability to bundle nested data into a single call provides a substantial advantage as complexity increases. For the five-call sequence, we observed an approximate 4x reduction in RTT for the Pet Clinic and Task Management projects, with the other three projects showing a 2-3x RTT reduction.
This performance gain arises entirely due to the reduction in network round-trip. Since the GraphQLify-generated API and the original REST API execute the same underlying business logic, the server-side computation time will be slightly higher for GraphQL due to framework overhead. We did not quantify this overhead as it is negligible compared to the network round-trip cost, which dominates overall latency.

\begin{figure}
  \centering
  \includegraphics[width=\linewidth]{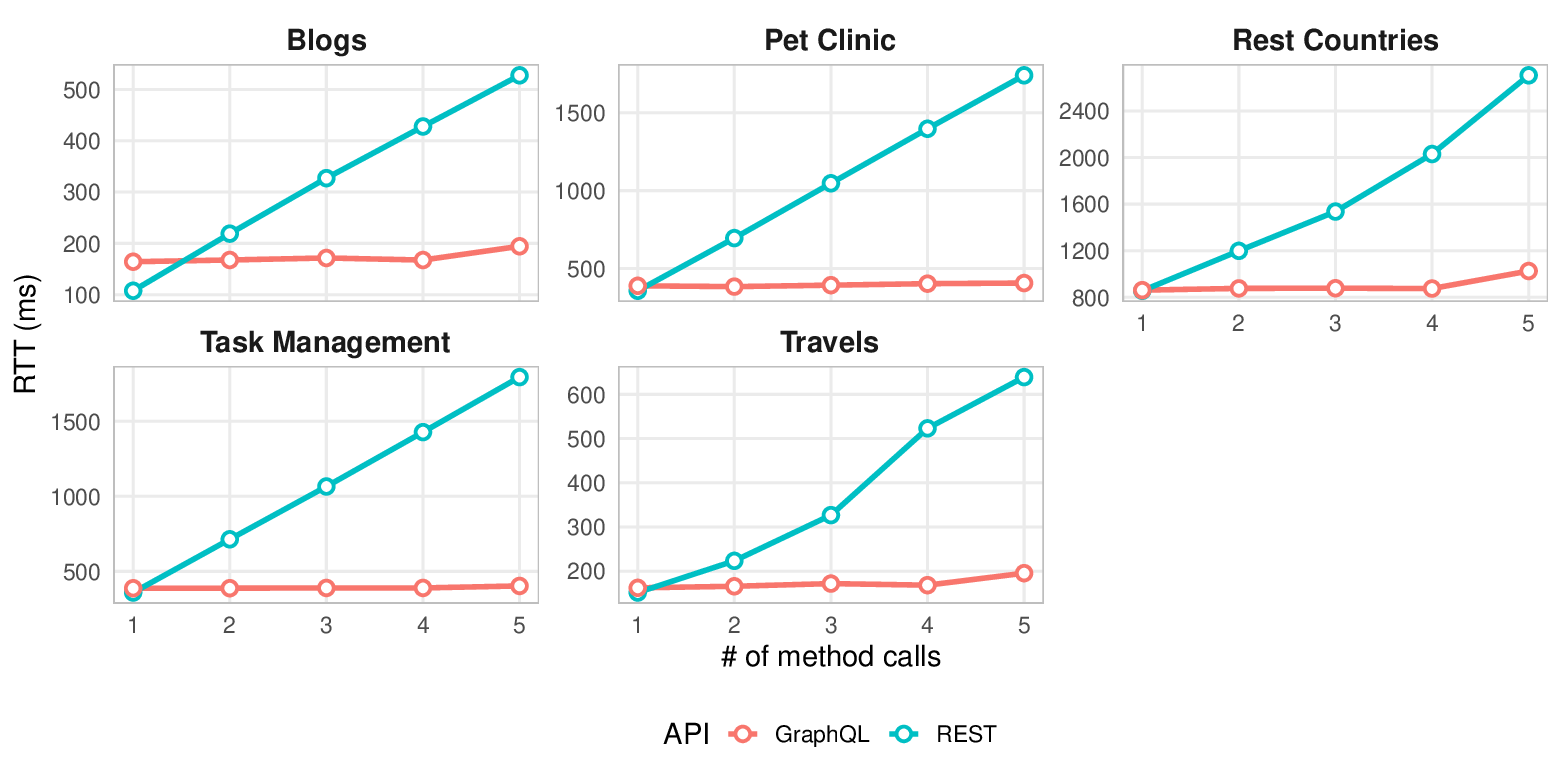}
      \caption{Comparison of data fetching time between REST and GraphQL. The X-axis shows the number of method calls required for a task. The Y-axis shows the average RTT time over five repetitions. }
  \label{fig:graphqlexample}  
\end{figure}

\begin{boxedtext}
    \textbf{Evaluation Summary:}  GraphQLify successfully converted 834 REST APIs across nine projects to GraphQL with perfect type fidelity. Unlike the state-of-the-art tool OASGraph (which had a 3.5\% failure rate and 42\% type mismatch), GraphQLify achieved zero conversion failures or type mismatches. This automated conversion was further verified by a manual validation of 50 APIs, which confirmed a flawless end-to-end migration. Beyond correctness, the generated APIs deliver significant performance gains. For application workflows requiring five sequential API calls, clients can reduce data fetching time by a factor of 2 to 4 by using the GraphQLify-generated APIs over REST counterparts.
\end{boxedtext}

%% file: tables/evaluation-projects2.tex
\begin{table}
\centering
\caption{Performance of \textit{GraphQLify} in terms of five evaluation metrics with nine open source Java projects }
\label{table:evaluation_projects}
\resizebox{\linewidth}{!}{
\begin{tabular}{|p{3.2cm}|R{.8cm}|R{.8cm}|R{.8cm}|R{1cm}|R{1cm}|R{.8cm}|R{1.8cm}|R{1.5cm}|R{1.2cm}|R{1.5cm}|} \hline
    \multirow{2}{3.2cm}{\textbf{Project}}& \multirow{2}{0.8cm}{\textbf{Forks}} & \multirow{2}{0.8cm}{\textbf{Stars}} & \multirow{2}{0.8cm}{\textbf{APIs}} & \multirow{2}{1cm}{\textbf{LOC}}
    &    \multicolumn{2}{c|}{\textbf{Failure rate (\%)}}&  \multirow{2}{1.8cm}{\textbf{Invalid schema (\%)}}&  \multirow{2}{1.5cm}{\textbf{Schema size (LOC)}} &	 \multirow{2}{1.2cm}{\textbf{Type count}}	& \multirow{2}{1.5cm}{\textbf{Type mismatch}} \\ \hhline{~~~~~--~~~~} 
      & & & & & 	\textbf{{strict}}	& \textbf{\textit{non-strict}} & &  & &	 \\ \hline
\href{https://github.com/macrozheng/mall}{Mall} &  28.6K & 76.8K & 160 & 133.5K &  62.5 &	0	& 0& 	1,600 &	173 &	0  \\
\hline

\href{https://github.com/metasfresh/metasfresh}{Metas Fresh} & 577 & 1.7K & 439 & 3.4M  &39.8 	 & 0&	0	&5,949 &	956	& 0 \\
\hline

\href{https://github.com/exadel-inc/CompreFace}{Exadel CompareFace} & 687 & 5K & 64 & 19.7K &	10.9 &	0	&0	&528 &	115&	0 \\
\hline

\href{https://github.com/apilayer/restcountries}{REST Countries} & 355 & 2.2K & 27 & 2.6K & 3.7 &	0&	0	&85 	&15&	0 \\
\hline

\href{https://github.com/osopromadze/Spring-Boot-Blog-REST-API}{Blog REST API} & 333 & 728 & 54 & 5.2K &	0&	0	&0	&547 &	64&	0 \\
\hline

\href{https://github.com/vatri/spring-rest-ecommerce}{Ecommerce REST API} & 165 & 327 & 21 & 2.2K & 4.7&	0	&0&	249 &	42&	0 \\
\hline

\href{https://github.com/spring-petclinic/spring-petclinic-rest}{Pet Clinic REST} & 857 & 458 & 36 & 13.5K &	0	&0&	0	&204& 33&	0 \\
\hline

\href{https://github.com/mariazevedo88/travels-java-api}{Travels Java API} & 123 & 243 & 13 & 4.1K &	15.4&	0	&0&	160 &39	&0 \\
\hline

\href{https://github.com/ayushman1024/TASK-Management-System}{Task Mgmt System} & 50 & 89 & 19 & 2.6K &	21.1	&0	&0	&168& 34	&0\\ 
\hline
\multicolumn{5}{|r|}{\textbf{Overall}}  &	31.3	&0	&0	& - & - & - \\ \hline

\end{tabular}
}
\end{table}

%% file: Sections/discussion.tex
\section{Discussion and Limitations}
\label{sec:discussion}

The following subsections describe the implications and limitations of \textit{GraphQLify}.

\subsection{Need for Static Analysis-based Approaches:} Type safety is a core feature of GraphQL~\cite{quina2023systematicmapping}. Our results suggest that even a state-of-the-art tool such as OASGraph largely fails to preserve original types after conversion, violating this key principle. We acknowledge that an adapter-based approach targeting machine-parsable specification may be the only viable option when source code is not available. However, the most likely use case for such a tool is when developers would willingly convert their REST APIs to GraphQL. In such cases, source code is available and type safety must be a priority (e.g., enterprise microservices). Our evaluation shows both the feasibility and the superiority of employing a static-analysis-based approach for type-preserving conversion.

\vspace{-2pt}
\subsection{Framework Extensibility}
\textit{GraphQLify} approach is designed to be extensible for various programming languages and API frameworks. 
First, the plugin-based architecture of \textit{GraphQLify} is analogous to the GraphQL ecosystem, which has proven successful despite the substantial tooling required to support various languages and platforms. Second, the \textit{GraphQLify}'s \textit{Translator} module is programming-language-agnostic and does not require modifications to support new languages. While a \textit{Generator} must be implemented to support a new language, all popular languages have numerous mature tools to expedite the development of new generators. For example,  we can adapt the static analysis capability of the \textit{GraphQLify} processor by utilizing stable static analysis tools supporting other programming languages such as C\#, Python, Go, or Ruby.
Although we need to build framework plugins to interpret Domain Specific Languages (DSLs) or frameworks commonly used to define APIs to handle the unique way each framework defines RESTful APIs and HTTP methods, the required efforts would be low. For example, the amount of code to develop our plugins for Spring, JAX-RS, and OpenAPI APIs was 205, 129, and 128 LOC.

\vspace{-2pt}
\subsection{Future Research Directions}
The GraphQL specification does not address authentication and authorization because it is designed to be a flexible query language and execution engine used in various contexts and with different types of data sources and recommends this aspect to be decided in the business logic layer rather than at resolvers~\cite{graphql_auth}.

\textit{GraphQLify}'s authorization leverages the ones already implemented in the APIs (e.g., REST) that it migrates. Future works may explore strategies to integrate Authentication and granular Authorization in \textit{GraphQLify}. For example, one can use custom annotations for authorization at the resolver level. A unified auth gateway protecting the data layer is another possible direction when multiple microservices may be using the same data source. Hence, it would be interesting to design a unified approach to handle all these diverse cases.

%% file: Sections/limitations.tex
\vspace{-2pt}
\subsection{Limitations} 
\label{sec:limitations}
\textit{(First)}, \textit{GraphQLify} uses static analysis of source code to generate a GraphQL schema. Therefore, known limitations of static analysis, such as failure to identify wildcard generic arguments, may lead to incomplete or inaccurate GraphQL schema or impact the generalizability of our findings.
\textit{(Second)}, \textit{GraphQLify}'s \code{non-strict} mode applies heuristics to mitigate conflicts. Since these rules are context-agnostic, they may not reflect the intended purposes and semantics of existing APIs. 
\textit{(Third)}, since our prototype implementation of \textit{GraphQLify} targets Java, its reported performance may not be generalizable to other programming languages.  Type inference using static analysis is challenging for programming languages using dynamic code generation, metaprogramming, or dynamic typing (e.g., JavaScript). Hence,  \textit{GraphQLify}'s approach may lead to an incomplete or incorrect schema for systems using such languages. Regardless, Java is the most popular choice to build enterprise applications~\cite{enterprise-list}, and  GraphQL is expected to see the most adoption in this spectrum~\cite{ibm_gartner_blog_post}. 
\textit{(Fourth)}, while our evaluation incorporates a substantial sample size of 834 APIs from a diverse group of nine OSS projects, it may not represent the full spectrum of real-world applications sufficient to draw robust statistical conclusions, especially when considering edge cases and less common API patterns.
Finally, \textit{(Fifth)}, our evaluation demonstrates how \textit{GraphQLify}-generated APIs overcome the underfetching problem, but does not illustrate their effectiveness in addressing overfetching. This is because the overfetching problem is inherently client-dependent, occurring only when the server returns more data than a client requires. Since GraphQL queries give clients fine-grained control over requested fields, overfetching is mitigated by design; therefore, we did not conduct additional experiments to demonstrate it.

%% file: Sections/conclusion.tex
\section{Conclusion} \label{sec:conclusion}
This paper introduces \textit{GraphQLify}, an automated framework for facilitating the migration of existing APIs to GraphQL. While prior works in this area take relational databases, resource description frameworks (RDF), or machine-parsable API specifications as inputs, \textit{GraphQLify} takes a new approach by leveraging static source code analysis.  During our evaluation with 834 APIs from nine popular OSS projects, GraphQLify successfully converts all, while matching types from the original APIs.
While OASGraph, the current state-of-the-art, had 3.5\% failure rates and 42\% type mismatches in our evaluation, both numbers are at 0\% for GraphQLify.
 For application workflows requiring five sequential API calls, clients can reduce data fetching time by a factor of 2 to 4 by using the GraphQLify-generated APIs over REST counterparts.

%% file: references.bib
@article{DBLP:journals/loplas/Landi92,
  title={Undecidability of static analysis}, volume={1}, ISSN={1557-7384}, url={http://dx.doi.org/10.1145/161494.161501}, DOI={10.1145/161494.161501}, number={4}, journal={ACM Letters on Programming Languages and Systems}, publisher={Association for Computing Machinery (ACM)}, author={Landi, William}, year={1992}, month=dec, pages={323–337}
}

@article{jaxrs,
  title={Jax-rs: Java™ api for restful web services},
  author={Pericas-Geertsen, Santiago and Potociar, Marek},
  journal={Oracle Corporation},
  pages={1--84},
  year={2013}
}

@inproceedings{DBLP:conf/uss/MachirySCSKV17,
  author       = {Aravind Machiry and
                  Chad Spensky and
                  Jake Corina and
                  Nick Stephens and
                  Christopher Kruegel and
                  Giovanni Vigna},
  editor       = {Engin Kirda and
                  Thomas Ristenpart},
  title        = {{DR.} {CHECKER:} {A} Soundy Analysis for Linux Kernel Drivers},
  booktitle    = {26th {USENIX} Security Symposium, {USENIX} Security 2017, Vancouver,
                  BC, Canada, August 16-18, 2017},
  pages        = {1007--1024},
  publisher    = {{USENIX} Association},
  year         = {2017}
}

@inproceedings{DBLP:conf/ccs/RahamanXASTFKY19,
series={CCS ’19}, title={CryptoGuard: High Precision Detection of Cryptographic Vulnerabilities in Massive-sized Java Projects}, url={http://dx.doi.org/10.1145/3319535.3345659}, DOI={10.1145/3319535.3345659}, booktitle={Proceedings of the 2019 ACM SIGSAC Conference on Computer and Communications Security}, publisher={ACM}, author={Rahaman, Sazzadur and Xiao, Ya and Afrose, Sharmin and Shaon, Fahad and Tian, Ke and Frantz, Miles and Kantarcioglu, Murat and Yao, Danfeng (Daphne)}, year={2019}, month=nov, pages={2455–2472}, collection={CCS ’19}
  }

@inproceedings{mcfadden2024wendigo,
  title={WENDIGO: Deep Reinforcement Learning for Denial-of-Service Query Discovery in GraphQL}, url={http://dx.doi.org/10.1109/spw63631.2024.00012}, DOI={10.1109/spw63631.2024.00012}, booktitle={2024 IEEE Security and Privacy Workshops (SPW)}, publisher={IEEE}, author={McFadden, Shae and Maugeri, Marcello and Hicks, Chris and Mavroudis, Vasilios and Pierazzi, Fabio}, year={2024}, month=may, pages={68–75} 
}

@inproceedings{mavroudeas-ASE-2021,
  title={Learning GraphQL Query Cost}, url={http://dx.doi.org/10.1109/ase51524.2021.9678513}, DOI={10.1109/ase51524.2021.9678513}, booktitle={2021 36th IEEE/ACM International Conference on Automated Software Engineering (ASE)}, publisher={IEEE}, author={Mavroudeas, Georgios and Baudart, Guillaume and Cha, Alan and Hirzel, Martin and Laredo, Jim A. and Magdon-Ismail, Malik and Mandel, Louis and Wittern, Erik}, year={2021}, month=nov, pages={1146–1150} 
}

@inproceedings{cha-FSE-2020,
series={ESEC/FSE ’20}, 
title={A principled approach to GraphQL query cost analysis}, 
url={http://dx.doi.org/10.1145/3368089.3409670}, 
DOI={10.1145/3368089.3409670}, 
booktitle={Proceedings of the 28th ACM Joint Meeting on European Software Engineering Conference and Symposium on the Foundations of Software Engineering}, 
publisher={ACM}, 
author={Cha, Alan and Wittern, Erik and Baudart, Guillaume and Davis, James C. and Mandel, Louis and Laredo, Jim A.}, 
year={2020}, 
month=nov, 
pages={257–268}, 
collection={ESEC/FSE ’20}
}

@online{openapi, 
    title={OpenAPI Inititive}, 
    url={https://www.openapis.org/}, 
    lastaccessed={2026-04-08}
 }

@online{swagger, 
    title={Swagger, API Development forEveryone}, 
    url={https://swagger.io/}, lastaccessed={2026-04-08}
}

@online{springrest, 
    title={Building REST services with Spring}, 
    url={https://spring.io/guides/tutorials/rest}, lastaccessed={2026-04-08}
}

@online{maven, 
    title={Welcome to Apache Maven}, 
    url={https://maven.apache.org/}, lastaccessed={2026-04-08}
}

@online{postgraphile, 
    title={PostGraphile, Instant GraphQL API}, 
    url={https://www.graphile.org/postgraphile/}, lastaccessed={2026-04-08}
}

@online{hasura, 
    title={Hasura, Instant GraphQL APIs on your data}, 
    url={https://hasura.io/}, lastaccessed={2026-04-08}
}

@online{dgraph, 
    title={DGraph}, 
    url={https://dgraph.io/}, lastaccessed={2026-04-08}
}

@online{graphqlmesh, 
    title={GraphQL-Mesh}, 
    url={https://the-guild.dev/graphql/mesh}, lastaccessed={2026-04-08}
}

@online{prisma, 
    title={Prisma}, 
    url={https://www.prisma.io/}, lastaccessed={2026-04-08}
}

@online{javaAnnotationProcessor, 
    title={Package javax.annotation.processing}, 
    url={https://docs.oracle.com/javase/8/docs/api/javax/annotation/processing/package-summary.html}, lastaccessed = {2026-04-08}
}

@online{ibm_gartner_blog_post, 
    title={Seven key insights on GraphQL trends}, 
    author ={Derks, Roy},
    howpublished={https://www.ibm.com/blog/seven-key-insights-on-graphql-trends/},
    year={2023}, lastaccessed = {2026-04-08}
}

@online{graphql_auth, 
    title={GraphQL Authorization}, 
    url={https://graphql.org/learn/authorization/}, lastaccessed = {2026-04-08}
}

@online{graphql_spec, 
    title={GraphQL Specification}, 
    year={2021},
    url={https://spec.graphql.org/October2021/},lastaccessed = {2026-04-08}
}

@online{graphql_java, 
    title={GraphQL Java - The Java implementation of GraphQL}, 
    url={https://www.graphql-java.com/}, lastaccessed = {2026-04-08}
}

@inproceedings{graphqlwrappers,
  title={Generating GraphQL-Wrappers for REST(-like) APIs}, ISBN={9783319916620}, ISSN={1611-3349}, url={http://dx.doi.org/10.1007/978-3-319-91662-0_5}, DOI={10.1007/978-3-319-91662-0_5}, booktitle={Web Engineering}, publisher={Springer International Publishing}, author={Wittern, Erik and Cha, Alan and Laredo, Jim A.}, year={2018}, pages={65–83}
}

@inproceedings{controlledexperiment,
   title={REST vs GraphQL: A Controlled Experiment}, url={http://dx.doi.org/10.1109/icsa47634.2020.00016}, DOI={10.1109/icsa47634.2020.00016}, booktitle={2020 IEEE International Conference on Software Architecture (ICSA)}, publisher={IEEE}, author={Brito, Gleison and Valente, Marco Tulio}, year={2020}, month=mar, pages={81–91} 
}

@InProceedings{wittern-2019,
title={An Empirical Study of GraphQL Schemas}, ISBN={9783030337025}, ISSN={1611-3349}, url={http://dx.doi.org/10.1007/978-3-030-33702-5_1}, DOI={10.1007/978-3-030-33702-5_1}, booktitle={Service-Oriented Computing}, publisher={Springer International Publishing}, author={Wittern, Erik and Cha, Alan and Davis, James C. and Baudart, Guillaume and Mandel, Louis}, year={2019}, pages={3–19}
}

@inproceedings{lambers2024taint,
  title={Taint Analysis for Graph APIs Focusing on Broken Access Control}, volume={Volume 22, Issue 1}, ISSN={1860-5974}, url={http://dx.doi.org/10.46298/lmcs-22(1:18)2026}, DOI={10.46298/lmcs-22(1:18)2026}, journal={Logical Methods in Computer Science}, publisher={Centre pour la Communication Scientifique Directe (CCSD)}, author={Lambers, Leen and Sakizloglou, Lucas and Khakharova, Taisiya and Orejas, Fernando}, year={2026}, month=mar
}

@article{belhadi2024random,
  title={Random Testing and Evolutionary Testing for Fuzzing GraphQL APIs}, volume={18}, ISSN={1559-114X}, url={http://dx.doi.org/10.1145/3609427}, DOI={10.1145/3609427}, number={1}, journal={ACM Transactions on the Web}, publisher={Association for Computing Machinery (ACM)}, author={Belhadi, Asma and Zhang, Man and Arcuri, Andrea}, year={2024}, month=jan, pages={1–41}
}

@inproceedings{zetterlund2022harvesting,
  title={Harvesting Production GraphQL Queries to Detect Schema Faults}, url={http://dx.doi.org/10.1109/icst53961.2022.00014}, DOI={10.1109/icst53961.2022.00014}, booktitle={2022 IEEE Conference on Software Testing, Verification and Validation (ICST)}, publisher={IEEE}, author={Zetterlund, Louise and Tiwari, Deepika and Monperrus, Martin and Baudry, Benoit}, year={2022}, month=apr, pages={365–376}
}

@inproceedings{roksela2020evaluating,
   title={Evaluating execution strategies of GraphQL queries}, url={http://dx.doi.org/10.1109/tsp49548.2020.9163501}, DOI={10.1109/tsp49548.2020.9163501}, booktitle={2020 43rd International Conference on Telecommunications and Signal Processing (TSP)}, publisher={IEEE}, author={Roksela, Piotr and Konieczny, Marek and Zielinski, Slawomir}, year={2020}, month=jul, pages={640–644}
}

@online{graphqlfoundation, 
    title={GraphQL Foundation}, 
    url={https://graphql.org/community/foundation/},
    lastaccessed = {2026-04-08}
}

@online{stepzen, 
    title={StepZen, an IBM Company}, 
    url={https://dashboard.ibm.stepzen.com/},
    lastaccessed = {2026-04-08}
}

@online{spring, 
    title={Spring Framework Documentation}, 
    url={https://docs.spring.io/spring-framework/reference/index.html},
    lastaccessed = {2026-04-08}
}

@online{swaggertographal, 
    title={Swagger-to-GraphQL}, 
    url={https://github.com/yarax/swagger-to-graphql},
    lastaccessed = {2026-04-08}
}

@online{apollo, 
    title={Apollo GraphQL}, 
    url={https://www.apollographql.com/}
}

@inproceedings{graphqlformaldef2,
  series={WWW ’18}, title={Semantics and Complexity of GraphQL}, url={http://dx.doi.org/10.1145/3178876.3186014}, DOI={10.1145/3178876.3186014}, booktitle={Proceedings of the 2018 World Wide Web Conference on World Wide Web  - WWW ’18}, publisher={ACM Press}, author={Hartig, Olaf and Pérez, Jorge}, year={2018}, pages={1155–1164}, collection={WWW ’18}
}

@inproceedings{farre2019graphql,
  title={GraphQL Schema Generation for Data-Intensive Web APIs}, ISBN={9783030320652}, ISSN={1611-3349}, url={http://dx.doi.org/10.1007/978-3-030-32065-2_13}, DOI={10.1007/978-3-030-32065-2_13}, booktitle={Model and Data Engineering}, publisher={Springer International Publishing}, author={Farré, Carles and Varga, Jovan and Almar, Robert}, year={2019}, pages={184–194}
}

@inproceedings{taelman2018graphql,
  title={GraphQL-LD: linked data querying with GraphQL},
  author={Taelman, Ruben and Vander Sande, Miel and Verborgh, Ruben},
  booktitle={ISWC2018, the 17th International Semantic Web Conference},
  pages={1--4},
  year={2018}
}

@inproceedings{ugql,
  title={Automatic bootstrapping of GraphQL endpoints for RDF triple stores},
  author={Gleim, Lars Christoph and Holzheim, Tim and Koren, Istv{\'a}n and Decker, Stefan},
   booktitle={International Semantic Web Conference (ISWC) 2020},
  pages={119--134},
  year={2020}
}

@inproceedings{evolutionarytesting,
 series={GECCO ’22}, title={Evolutionary-based automated testing for GraphQL APIs}, url={http://dx.doi.org/10.1145/3520304.3528952}, DOI={10.1145/3520304.3528952}, booktitle={Proceedings of the Genetic and Evolutionary Computation Conference Companion}, publisher={ACM}, author={Belhadi, Asma and Zhang, Man and Arcuri, Andrea}, year={2022}, month=jul, pages={778–781}, collection={GECCO ’22}
}

@article{ontology,
 title={Ontology-based GraphQL server generation for data access and data integration}, volume={15}, ISSN={2210-4968}, url={http://dx.doi.org/10.3233/sw-233550}, DOI={10.3233/sw-233550}, number={5}, journal={Semantic Web: – Interoperability, Usability, Applicability}, publisher={SAGE Publications}, author={Li, Huanyu and Hartig, Olaf and Armiento, Rickard and Lambrix, Patrick}, year={2024}, month=jan, pages={1639–1675}
}

@inproceedings{seifer2019empirical,
 series={SLE ’19}, title={Empirical study on the usage of graph query languages in open source Java projects}, url={http://dx.doi.org/10.1145/3357766.3359541}, DOI={10.1145/3357766.3359541}, booktitle={Proceedings of the 12th ACM SIGPLAN International Conference on Software Language Engineering}, publisher={ACM}, author={Seifer, Philipp and Härtel, Johannes and Leinberger, Martin and Lämmel, Ralf and Staab, Steffen}, year={2019}, month=oct, pages={152–166}, collection={SLE ’19}
}

@online{altair,
    url ={https://altairgraphql.dev/} ,
    title = {Altair GraphQL Client},
    lastaccessed={2026-04-01}
}

@article{yazdipour2020github,
  title={GitHub data exposure and accessing blocked data using the GraphQL security design flaw},
  author={Yazdipour, Shahriar},
  journal={arXiv preprint arXiv:2005.13448},
  year={2020}
}

@misc{enterprise-list,
    title = {Enterprise Programming Languages for 2024},
    howpublished = {https://tridenstechnology.com/enterprise-programming-languages/},
   year={2024},
   author={Lesjak, Ziga },
}

@inproceedings{migratingtographql_apracticalassessment,
  title={Migrating to GraphQL: A Practical Assessment}, url={http://dx.doi.org/10.1109/saner.2019.8667986}, DOI={10.1109/saner.2019.8667986}, booktitle={2019 IEEE 26th International Conference on Software Analysis, Evolution and Reengineering (SANER)}, publisher={IEEE}, author={Brito, Gleison and Mombach, Thais and Valente, Marco Tulio}, year={2019}, month=feb, pages={140–150}
}

@inproceedings{replacerest,
  title={Can GraphQL Replace REST? A Study of Their Efficiency and Viability}, url={http://dx.doi.org/10.1109/ser-ip52554.2021.00009}, DOI={10.1109/ser-ip52554.2021.00009}, booktitle={2021 IEEE/ACM 8th International Workshop on Software Engineering Research and Industrial Practice}, publisher={IEEE}, author={Vadlamani, Sri Lakshmi and Emdon, Benjamin and Arts, Joshua and Baysal, Olga}, year={2021}, month=jun
}

@article{quina2023systematicmapping,
 title={GraphQL: A Systematic Mapping Study}, volume={55}, ISSN={1557-7341}, url={http://dx.doi.org/10.1145/3561818}, DOI={10.1145/3561818}, number={10}, journal={ACM Computing Surveys}, publisher={Association for Computing Machinery (ACM)}, author={Quiña-Mera, Antonio and Fernandez, Pablo and García, José María and Ruiz-Cortés, Antonio}, year={2023}, month=feb, pages={1–35} 
}

@article{morphgraphql,
 title={Exploiting Declarative Mapping Rules for Generating GraphQL Servers with Morph-GraphQL}, volume={30}, ISSN={1793-6403}, url={http://dx.doi.org/10.1142/s0218194020400070}, DOI={10.1142/s0218194020400070}, number={06}, journal={International Journal of Software Engineering and Knowledge Engineering}, publisher={World Scientific Pub Co Pte Lt}, author={Chaves-Fraga, David and Priyatna, Freddy and Alobaid, Ahmad and Corcho, Oscar}, year={2020}, month=jun, pages={785–803}
}

@inproceedings{quina2021quality,
  title={Quality in Use Evaluation of a GraphQL Implementation}, ISBN={9783030960438}, ISSN={2367-3389}, url={http://dx.doi.org/10.1007/978-3-030-96043-8_2}, DOI={10.1007/978-3-030-96043-8_2}, booktitle={Emerging Research in Intelligent Systems}, publisher={Springer International Publishing}, author={Quiña-Mera, Antonio and Fernández-Montes, Pablo and García, José María and Bastidas, Edwin and Ruiz-Cortés, Antonio}, year={2022}, pages={15–27}
}

@article{spoon2006,
 title={<scp>SPOON</scp>: A library for implementing analyses and transformations of Java source code: SPOON}, volume={46}, ISSN={0038-0644}, url={http://dx.doi.org/10.1002/spe.2346}, DOI={10.1002/spe.2346}, number={9}, journal={Software: Practice and Experience}, publisher={Wiley}, author={Pawlak, Renaud and Monperrus, Martin and Petitprez, Nicolas and Noguera, Carlos and Seinturier, Lionel}, year={2015}, month=aug, pages={1155–1179}
}
